\documentclass[conference]{IEEEtran}
\IEEEoverridecommandlockouts

\usepackage{amsmath,amssymb,amsfonts} 
\usepackage{amsthm}                   

\usepackage{graphicx}    
\usepackage{subfigure}   
\usepackage{subcaption}  
\usepackage{booktabs}    
\usepackage{multirow}    
\usepackage{diagbox}     
\usepackage{makecell}    
\usepackage{float}       
\usepackage{balance}     
\usepackage{multicol}    
\usepackage{textcomp}    

\usepackage[linesnumbered, lined, boxed, commentsnumbered, ruled, vlined]{algorithm2e} 
\usepackage{algorithmic} 

\usepackage{xcolor}      
\usepackage{color}       
\usepackage{microtype}   
\usepackage{enumitem}    
\usepackage{comment}     
\usepackage{cite}        

\usepackage{url}  
\usepackage[colorlinks=true, allcolors=black]{hyperref}

\usepackage{amsthm}
\theoremstyle{plain}
\newtheorem{definition}{Definition}
\newtheorem{theorem}{Theorem}

\renewenvironment{proof}[1][Proof]{%
  \par\noindent\textbf{#1.} }%
  {\hfill$\square$\par}

\usepackage[linesnumbered, ruled, vlined]{algorithm2e}  
\newcommand\mycommfont[1]{\color{green!50!black}\ttfamily\small #1}
\SetCommentSty{\mycommfont}
\SetKwInOut{Input}{Input}
\SetKwInOut{Output}{Output}

\makeatletter
\renewcommand{\footnoterule}{%
    \kern -3pt 
    \hrule width 0.5\columnwidth height 0.4pt 
    \kern 2pt 
}
\makeatother

\usepackage{graphicx}      
\usepackage{float}         
\usepackage{subcaption}    
\usepackage{booktabs}      
\usepackage{multirow}      
\usepackage{makecell}      
\usepackage{balance}       
\usepackage{multicol}      
\usepackage{enumitem}      

\newcommand{\eat}[1]{}

\newcommand{\changecolor}[1]{}          

\usepackage[ruled,vlined]{algorithm2e}
\usepackage{xcolor}

\definecolor{mygreen}{rgb}{0.0, 0.5, 0.0}

\def\BibTeX{{\rm B\kern-.05em{\sc i\kern-.025em b}\kern-.08em
    T\kern-.1667em\lower.7ex\hbox{E}\kern-.125emX}}
\begin{document}

\title{
FLASH Viterbi: Fast and Adaptive Viterbi Decoding for Modern Data Systems\\
\thanks{$\dagger$ These authors contributed equally to this work. *Corresponding author.}
}
\author{
    \IEEEauthorblockN{
        Ziheng Deng$^{1,\dagger}$, Xue Liu$^{1,\dagger}$, Jiantong Jiang$^{2}$, Yankai Li$^{1}$, Qingxu Deng$^{1,*}$, Xiaochun Yang$^{1}$
    }
    \IEEEauthorblockA{
        \textit{$^1$School of CSE, Northeastern University, China} 
        \textit{$^2$School of CIS, The University of Melbourne, Australia} \\
        $^1$\{2310721@stu, liuxue@cse, 20215828@stu, dengqx@mail, yangxc@mail\}.neu.edu.cn, $^2$ jiantong.jiang@unimelb.edu.au
    }
}
\maketitle
\begin{abstract}

The Viterbi algorithm is a key operator for structured sequence inference in modern data systems, with applications in trajectory analysis, online recommendation, and speech recognition. As these workloads increasingly migrate to resource-constrained edge platforms, standard Viterbi decoding remains memory-intensive and computationally inflexible. Existing methods typically trade decoding time for space efficiency, but often incur significant runtime overhead and lack adaptability to various system constraints. 
This paper presents \textsc{FLASH Viterbi}, a Fast, Lightweight, Adaptive, and Hardware-Friendly Viterbi decoding operator that enhances adaptability and resource efficiency. \textsc{FLASH Viterbi} combines a non-recursive divide-and-conquer strategy with pruning and parallelization techniques to enhance both time and memory efficiency, making it well-suited for resource-constrained data systems. 
To further decouple space complexity from the hidden state space size, we present \textsc{FLASH-BS Viterbi}, a dynamic beam search variant built on a memory-efficient data structure. Both proposed algorithms exhibit strong adaptivity to diverse deployment scenarios by dynamically tuning internal parameters.
To ensure practical deployment on edge devices, we also develop FPGA-based hardware accelerators for both algorithms, demonstrating high throughput and low resource usage.
Extensive experiments show that our algorithms consistently outperform existing baselines in both decoding time and memory efficiency, while preserving adaptability and hardware-friendly characteristics essential for modern data systems. All codes are publicly available at \url{https://github.com/Dzh-16/FLASH-Viterbi}.
\end{abstract}

\section{Introduction}
    \begin{figure}
		\centerline{\includegraphics[width=1\linewidth]{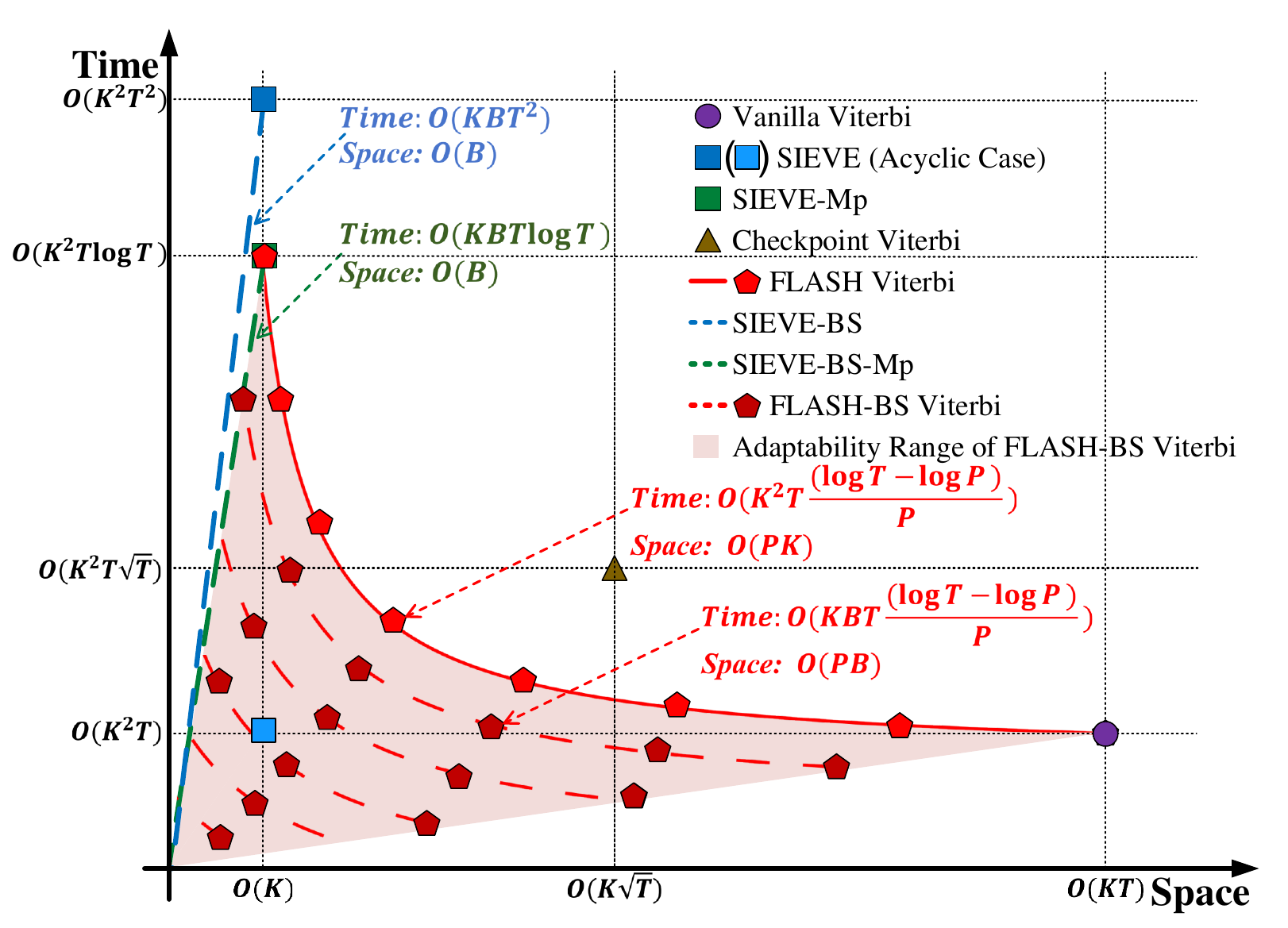}}
		\caption{Theoretical time-space complexity comparison of \textsc{FLASH Viterbi}, \textsc{FLASH-BS Viterbi}, and baseline algorithms. 
        $K$ is the HMM state space size, $T$ is the sequence length, $B$ is the beam width, and $P$ is the parallelism degree.
        }
		\label{Fig:Theoretical comparison}
    \vspace{-2mm}
	\end{figure}
Viterbi decoding~\cite{HMM_2} is a critical operator for real-time structured inference in data systems, particularly in edge environments with limited resources. As the core decoding algorithm in probabilistic sequence models and chain-structured probabilistic graphical models such as Hidden Markov Models (HMMs) and Conditional Random Fields (CRFs)~\cite{HMM, LHMM,AutomaticRoad,LRNI, Data-Driven, TextCorporaHMMs, GraphMM,jiang2024fastpgm, jiang2024fast, ResuFormer,Maxson,CRFS3, jiang2022fast}, it recursively computes the optimal state sequence by maximizing the joint likelihood over all possible paths. This capability enables a wide range of data-driven applications like mobility trajectory inference~\cite{LHMM,AutomaticRoad,LRNI,Data-Driven}, online recommendation~\cite{OSMRec}, behavioral labeling~\cite{IndoorMobility}, and speech recognition~\cite{SIEVE,SIEVE_bs,HMM_Speech,HMM_Speech_2}.

However, as inference tasks grow increasingly complex, the standard Viterbi algorithm  (i.e., Vanilla Viterbi) remains resource-intensive and ill-suited for edge deployment~\cite{Viterbihardwarerequirements}. Specifically, when applied to an HMM, it incurs a space complexity of $\mathcal{O}(KT)$ and a time complexity of $\mathcal{O}(K^2T)$, where $K$ is the hidden state size and $T$ is the observation sequence length, as illustrated in Figure~\ref{Fig:Theoretical comparison}. They scale linearly or quadratically with $K$ and $T$, hindering deployment on resource-limited edge platforms~\cite{Viterbihardwarerequirements,Viterbihardwarerequirements2,Yazdani2017}. Such resource demands have emerged as a challenge in modern serving systems~\cite{miao2025towards, jiang2025towards}.

Recent efforts have attempted to trade off decoding time and memory usage~\cite{Lazy_Viterbi,Viterbi_problem,checkpoint,SIEVE,SIEVE_bs}, yet still incur high memory overhead or significant latency. Importantly, most approaches lack adaptivity to system-level constraints and fail to optimize Viterbi under the dynamic resource limitations typical of modern data systems, especially on edge platforms.

To this end, we propose \textsc{FLASH Viterbi} and its beam search variant \textsc{FLASH-BS Viterbi}, a new family of fast, lightweight, adaptive, and hardware-friendly 
Viterbi algorithms that address the challenges of space and time efficiency jointly. We aim to support Viterbi as a modular operator within real-time processing pipelines on edge systems.
Specifically, \textsc{FLASH Viterbi} adopts a non-recursive divide-and-conquer strategy to decompose the decoding task into subtasks, improving memory reuse and reducing scheduling overhead compared to recursive methods. Moreover, this strategy enables large-scale parallelism by reordering subtask execution.
At the subtask level, a unified pruning and parallelization strategy eliminates redundant computations and interdependencies, thereby unlocking full parallelism. 
Building on this, \textsc{FLASH-BS Viterbi} integrates dynamic beam search with an efficient data structure tailored for frequent candidate path updates.
Both algorithms employ a double-buffered memory scheme to reduce the overhead of intermediate data exchange.

These designs collectively endow \textsc{FLASH Viterbi} and \textsc{FLASH-BS Viterbi} with \textit{fast} execution and \textit{lightweight} memory usage. More importantly, they offer strong \textit{adaptivity} to diverse deployment scenarios by dynamically tuning internal parameters such as parallelism degree $P$ and beam width $B$. 
This adaptivity is illustrated by the theoretical time–space complexity comparison shown in Figure~\ref{Fig:Theoretical comparison}. The red solid curve captures the performance range of \textsc{FLASH Viterbi} under different degrees of parallelism. 
By adjusting \( P \), \textsc{FLASH Viterbi} achieves a flexible trade-off between time and space complexity, as indicated by the bright red pentagons, enabling efficient resource utilization.
\textsc{FLASH-BS Viterbi} further expands this tunable space by adjusting the beam width \( B \), as shown by the red shaded region.
Subsequent experiments show that our method outperforms existing baselines in both time and space dimensions.
In addition to the software implementation, our FPGA-based accelerator further demonstrates the \textit{hardware-friendliness} of our design. Our method’s non-recursive structure and low memory footprint translate to better hardware utilization, making it suitable for real-time, energy-efficient deployment on edge devices. The key contributions of this paper can be summarized as follows.
\begin{itemize}[leftmargin=*,itemsep=2pt,topsep=3pt,parsep=0pt]
\item We propose \textsc{FLASH Viterbi}, a Viterbi decoding algorithm that employs non-recursive design, pruning, and parallelization to achieve high time and memory efficiency.
\item We introduce \textsc{FLASH-BS Viterbi}, a variant incorporating beam search to enhance adaptability. It dynamically tunes parallelism degree and beam width to balance latency and memory usage under varying resource constraints.
\item We demonstrate the hardware-friendliness of our design through an FPGA-based implementation, enabling real-time and energy-efficient deployment on edge devices.
\end{itemize}

\section{Related Work}

To enable real-time decoding on resource-constrained edge computing platforms, a variety of improved Viterbi algorithms have been proposed~\cite{Viterbi_problem}.

\subsection{Space-Efficient Viterbi Variants}
Tarnas et al.~\cite{checkpoint} introduced Checkpoint Viterbi, which stores optimal path state information at \( \sqrt{T} \) evenly spaced timesteps (checkpoints) and re-executes the algorithm between them. Although this reduces the space complexity to \( \mathcal{O}(K\sqrt{T}) \), it still scales with sequence length \( T \). Ciaperoni et al.~\cite{SIEVE} proposed SIEVE, a divide-and-conquer approach based on state-space decomposition. It retains only the midpoint state pairs from the optimal path that lie in the center of state transitions, and uses them to divide decoding. This strategy reduces space complexity to \( \mathcal{O}(K) \). Yet, locating such midpoints requires a BFS traversal, rendering SIEVE sensitive to the density of the state-transition graph. In cyclic graphs, the time complexity may increase to \( \mathcal{O}(K^2T^2) \). A variant, SIEVE-MiddlePath (SIEVE-Mp)~\cite{SIEVE}, adopts a divide-and-conquer strategy based on sequence length. It maintains a space complexity of \( \mathcal{O}(K) \) and consistently achieves \( \mathcal{O}(K^2T\log{T}) \) time complexity regardless of graph density, making it the SOTA approach for space-efficient Viterbi decoding. However, existing space-efficient Viterbi algorithms primarily rely on recursive implementations, which incur considerable overhead from stack memory and scheduling. In addition, the decoding order of subtasks often misaligns with their generation dependencies, hindering parallel execution.

\subsection{Beam Search-based Viterbi Enhancements}

Beam search is widely adopted for efficient Viterbi decoding~\cite{Beam_search_1,kaldi,Beam_search_2,Beam_search_3,Beam_search_4,Beam_search_5}. Instead of maintaining all possible paths for the full state space (\(K\)), it retains only the top \(B\) candidate paths (where \(B\) is the beam width) at each timestep. This strategy reduces both computational cost and memory footprint with minimal impact on accuracy.
Beam search implementations are generally categorized as static or dynamic based on when path pruning occurs. Static beam search either retains the top-\(B\) paths after computing all candidates~\cite{SIEVE_bs} or uses fixed thresholds for pruning~\cite{Yazdani2017,He,Beam_search_4}. However, such methods must temporarily store scores for all \(K\) states at each timestep, leading to limited actual memory savings.
In contrast, dynamic beam search~\cite{Michael,He2} performs real-time pruning during the computation of path probabilities, maintaining only \(B\) paths at each step. This design significantly improves memory efficiency and makes dynamic beam search more suitable for deployment on resource-constrained platforms. However, its effectiveness depends on efficient data structures that support frequent path updates and replacements during decoding.

\subsection{Hardware-Accelerated Viterbi Decoding}

Hardware acceleration has been actively explored to enhance Viterbi decoding~\cite{FPGA-basedViterbi,FPGA-basedViterbi2,Beam_search_1}. 
Pani et al.~\cite{ReducedViterbi} proposed an FPGA-based beam search Viterbi decoder that employs a memory architecture comprising two binary search trees and a max-heap to support path updates. Despite its effectiveness, the design exhibits structural redundancy, leading to high on-chip resource usage and increased power consumption.

\section{Preliminaries}\label{sec:sysmodel}

    In edge-side sequence inference, sensor data is typically represented as a sequence of observed features. 
    Hidden Markov Models (HMMs) offer a probabilistic framework that links these observations to an underlying chain of hidden states. 
    The states evolve according to a first-order Markov process, where each state depends only on its predecessor, 
    while each observation is generated from the emission distribution associated with the current state~\cite{murphy2012machine, miyato2015distributional, yang2024regulating}. 
    By modeling both the temporal dynamics of the hidden states and their influence on the observations, HMMs effectively capture dependencies in sensor data and allow efficient decoding via algorithms like Viterbi.
        
    A typical HMM consists of the following five key components:  
1) An observation set \( O \) containing \( M \) observable symbols;  
2) A hidden state set \( S \) comprising \( K \) hidden states;  
3) An initial state probability vector \( \pi \), specifying the prior distribution over hidden states at the initial timestep;  
4) A state transition probability matrix \(\mathcal{A}\), defining the probabilities of transitions between hidden states;  
5) An emission probability matrix \( \mathcal{B} \), representing the conditional probabilities of observations given hidden states. Hidden states are discrete, while observations may be either discrete (e.g., modeled by categorical distributions) or continuous (e.g., modeled by Gaussian mixture models or deep neural networks~\cite{Deep_Neural}). An HMM can be specified by the triplet \( (\pi, \mathcal{A}, \mathcal{B}) \).    

\begin{figure}[t!]
    \centering
    \includegraphics[width=1\columnwidth]{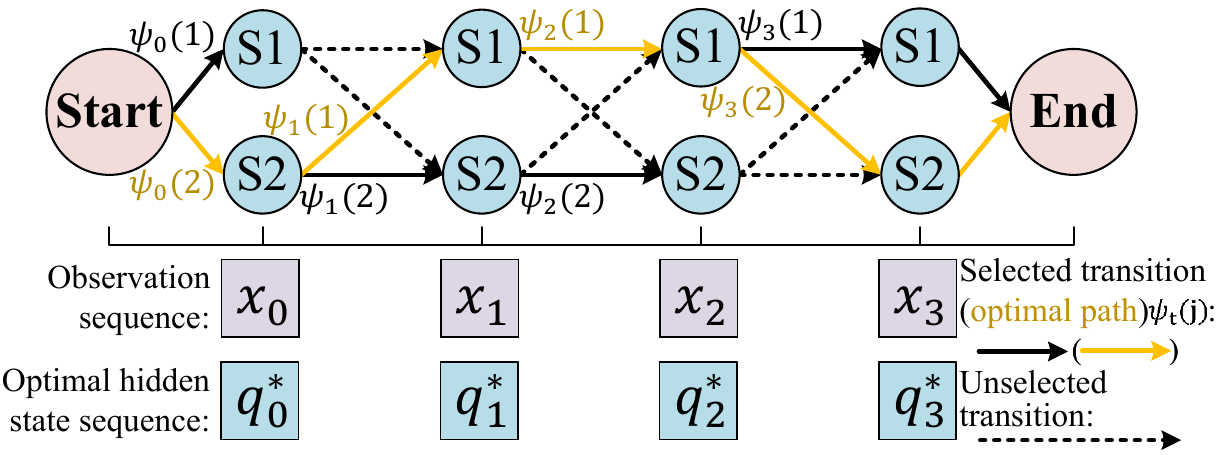}
    \caption{Trellis diagram of Viterbi decoding, showing the optimal hidden state sequence (\(q_t^*\)) for the observation sequence (\(x_t\)).}
    \label{fig:BACKGROUND}
    \vspace{-2mm}
\end{figure}

\subsection{Viterbi Algorithm}

Given an HMM $\lambda = (\pi, \mathcal{A}, \mathcal{B})$, the decoding problem is defined as follows: For an observation sequence \( X = \{x_0, x_1, \dots, x_{T-1}\} \), drawn from the observation set \( O = \{o_1, o_2, \dots, o_M\} \), determine the optimal hidden state sequence \( Q^* = \{q_0^*, q_1^*, \dots, q_{T-1}^*\} \), denoted as the optimal path, selected from the hidden state set \( S = \{s_1, s_2, \dots, s_K\} \), that satisfies:
{\setlength{\abovedisplayskip}{3pt}
 \setlength{\belowdisplayskip}{3pt}
\begin{align}
Q^* &= \underset{Q}{\arg\max}\, P(Q\,|\,X, \lambda), \notag 
\label{eq1}
\end{align}
}
where \(Q = \{q_0, q_1, \dots, q_{T-1}\} \) denotes all possible hidden state sequences of length \( T \). The Viterbi algorithm employs dynamic programming to solve HMM decoding problems. As illustrated in Figure~\ref{fig:BACKGROUND}, the decoding procedure relies on the recursive computation of two key variables: \( \delta_t(j) \),\( \psi_t(j) \)~\cite{viterbi_algorithm}:
{\setlength{\abovedisplayskip}{3pt}
 \setlength{\belowdisplayskip}{3pt}
\begin{align*}
        \delta_0(j) = \pi_j \cdot \mathcal{B}_j(x_0), \quad \psi_0(j) = 0 \quad \forall j \in \{1,...,K\} \quad (t=0);
\end{align*}
}
\vspace{-2.5mm}
{\setlength{\abovedisplayskip}{3pt}
 \setlength{\belowdisplayskip}{3pt}
\begin{equation*}
\begin{aligned}
    \delta_t(j) &= \max_{1 \leq i \leq K} \left[ \delta_{t-1}(i) \cdot \mathcal{A}_{ij} \right] \cdot \mathcal{B}_j(x_t) ,\\
    \psi_t(j) &= \underset{1 \leq i \leq K}{\arg\max} \left[ \delta_{t-1}(i) \cdot \mathcal{A}_{ij} \right] \quad (t=1,...,{T-1}).
\end{aligned}
\label{eq3}
\end{equation*}
}
Here, \( \delta_t(j) \) denotes the highest probability of any state sequence that accounts for the first \( t \) observations and ends in state \( j \) at timestep \( t \). The variable \( \psi_t(j) \) stores the index of the previous state that leads to the maximum \( \delta_t(j) \), which corresponds to the selected transition into state \( j \) shown by the solid lines in Figure~\ref{fig:BACKGROUND}, while the dashed lines indicate unselected transitions.
After completing the recursion, the optimal final state \( q^*_{T-1} \) is selected as the state that maximizes \( \delta_{T-1}(j) \). The complete path is then reconstructed by backtracking from \( q^*_{T-1} \) through \( \psi_t(j) \), as illustrated by the yellow solid path, enabling path-wise attribution~\cite{sundararajan2017axiomatic, erion2021improving, yang2023local}.
   {\setlength{\abovedisplayskip}{3pt}
 \setlength{\belowdisplayskip}{3pt}
\begin{align*}
        &q_{T-1}^* = \underset{1 \leq j \leq K}{\arg\max}\ \delta_{T-1}(j), \\
        &q_t^* = \psi_{t+1}(q_{t+1}^*) \quad \text{for } \quad (t={T-2},...,0).
\label{eq4}
\end{align*}
}

The space complexity of the Viterbi algorithm is \( \mathcal{O}(KT) \), as it stores \( \delta_t(j) \) and \( \psi_t(j) \) at each timestep. Its time complexity is \( \mathcal{O}(K^2T) \), due to evaluating all state transitions at each timestep. 
To prevent numerical overflow, \( \delta_t(j) \) and \( \psi_t(j) \) are computed using logarithmic summation rather than direct probability multiplications, as shown in lines 5, 8, 14, 15 of Algorithm~\ref{Alg:FLASH Viterb Decoding Process}.

\section{Overview of \textsc{FLASH Viterbi}}

To efficiently utilize limited resources on edge devices for Viterbi decoding, we propose \textsc{FLASH Viterbi}—a Fast, Lightweight, Adaptable, and Hardware-friendly Viterbi decoding algorithm. It incorporates the following core techniques.

\textit{1) Non-recursive Divide-and-conquer for Task Management.}
\textsc{FLASH Viterbi} pre-generates subtask sets and execution order based on inter-subtask dependencies. It employs a task queue for fine-grained control, reducing space complexity from $\mathcal{O}(KT)$ to $\mathcal{O}(K)$ while avoiding the resource overhead of recursive implementations. Moreover, the resulting execution plan supports parallel decoding of subtasks across threads.
We describe the details in Section~\ref{sec_tech1}.

\textit{2) Pruning and Parallelization for Decoding Efficiency.} 
During subtask decoding, \textsc{FLASH Viterbi} prunes candidate paths of subsequent subtasks based on the optimal state determined from the current subtask, retaining only transitions from that state.
This strategy provides two significant 
benefits: reducing redundant computations and removing dependencies among parallel subtasks. 
With parallelism degree $P$, the overall time complexity is reduced to $\mathcal{O}\left(\frac{K^2T(\log{T}-\log{P})}{P}\right)$. We elaborate our pruning and parallelization strategy in Section~\ref{sec_tech2}.

\textit{3) Dynamic Beam Search Integration.} 
By combining dynamic beam search with \textsc{FLASH Viterbi}, we introduce the variant \textsc{FLASH-BS Viterbi}. It retains 
only the top-$B$ candidate paths during decoding, which 
reduces space complexity to $\mathcal{O}(B)$ and decouples it from state space size $K$, 
\textsc{FLASH-BS Viterbi} offers strong adaptivity to diverse scenarios by configuring beam width.
We also design an efficient data structure to support 
inherent frequent updates. 
We describe the variant \textsc{FLASH-BS Viterbi} in Section~\ref{sec_tech3}.

\textit{4) Hardware-friendly Design.}
The proposed algorithms eliminate recursion and BFS operations from the execution flow, employ a double-buffered memory scheme in memory structures, and support configurable parallelism, making them well-suited for hardware acceleration. We implement FPGA-based accelerators for both \textsc{FLASH Viterbi} and \textsc{FLASH-BS Viterbi}, achieving high throughput with low resource consumption.  
The details of our hardware-accelerated implementation are provided in Section~\ref{sec_tech4}.

\section{\textsc{FLASH Viterbi} Design}

In this section, we present the design of \textsc{FLASH Viterbi} and its variant \textsc{FLASH-BS Viterbi}, both of which are optimized for time and memory efficiency. \textsc{FLASH Viterbi} combines a non-recursive divide-and-conquer strategy with pruning and parallelization techniques to enhance decoding performance and parallel adaptability. The variant \textsc{FLASH-BS Viterbi} integrates dynamic beam search and adapts to available computational resources, making it suitable for deployment on resource-constrained edge devices. In the following discussion, we describe the technical details.

\subsection{Non-recursive Divide-and-conquer Strategy}
\label{sec_tech1}
    \begin{figure}[tb]
		\centerline{\includegraphics[width=1\linewidth]{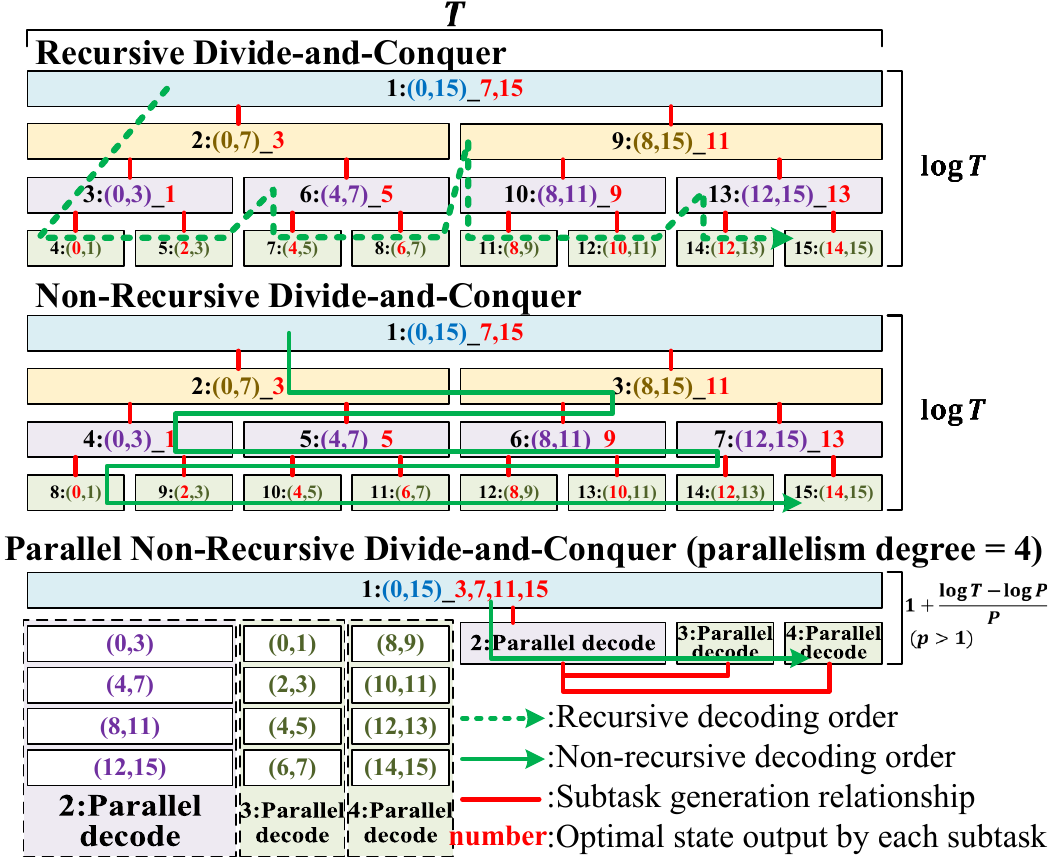}}
		\caption{Comparison of recursive, non-recursive, and parallel non-recursive divide-and-conquer Viterbi decoding. Each rectangle represents a decoding subtask. A subtask labeled $i: (a,b)\_t$ indicates the $i$-th subtask that decodes the segment $(a,b)$ and outputs the optimal path state at timestep~$t$, i.e., $q^*_t$.
        }
		\label{Fig:comparison of recursive}
            \vspace{-2mm}
	\end{figure}

To address the limitations of prior Viterbi algorithms on resource-constrained devices, \textsc{FLASH Viterbi} adopts a non-recursive divide-and-conquer strategy for task decomposition and decoding. This section discusses the limitations of recursive decoding and presents our alternative strategy that enables fine-grained task management and facilitates efficient parallelization.

\subsubsection{Limitations of Recursive Viterbi}
To reduce the memory footprint of the Viterbi algorithm, a common approach in data systems is to decompose the decoding task into multiple subtasks via a multi-stage divide-and-conquer process, where each subtask retains only limited state information from the optimal path for output. After decoding, subtasks are recursively divided into smaller subtasks to identify the remaining optimal states. 
Once all subtasks are decoded, the full optimal path is reconstructed by aggregating their outputs. Existing space-efficient Viterbi variants, such as SIEVE-Mp~\cite{SIEVE}, employ recursive strategy to manage the generate and decoding order of subtasks. However, this recursive formulation has two key limitations: First, the recursive structure inherently constrains parallelism across subtasks. As shown in Figure~\ref{Fig:comparison of recursive}, when the sequence length is 16, the red lines indicate the generation relationships between subtasks. Subtasks of the same color within a layer are derived from their parent subtasks in the previous layer and are independent of each other, theoretically enabling parallel decoding. However, the green dashed lines represent the recursive execution order, which conflicts with the generation dependencies. Specifically, subtasks within the same layer are executed non-consecutively and far apart, hindering parallel execution at that layer. Secondly, recursive decoding requires auxiliary stack space to preserve context from upper recursion levels. The stack depth scales with the sequence length \( T \), introducing additional memory overhead that undermines the objective of minimizing memory usage.

\subsubsection{Our Non-recursive Solution}

To address the limitations of recursive decoding, \textsc{FLASH Viterbi} pre-generates subtasks based on their generation dependencies. Each subtask is represented as a tuple \((t_{\text{start}}, t_{\text{end}})\), where \( t_{\text{start}} \) and \( t_{\text{end}} \) denote the start and end timesteps of a decoding segment, respectively. For a decoding task with sequence length \(T\), the initial tuple is defined as \((0, T{-}1)\). A task queue is employed to manage execution order through dynamic enqueue and dequeue operations, following two key principles: (i) Inter-layer ordering, where a parent subtask is always processed before its child subtasks to ensure the required state information is available to children (e.g., yellow tasks precede purple tasks in Figure~\ref{Fig:comparison of recursive}); 
(ii) Intra-layer parallelism, where subtasks within the same layer exhibit no generation dependencies. Such independence allows them to be decoded in any order to support parallel execution. For serial execution, subtasks are processed sequentially according to their start timesteps $t_{\text{start}}$. Figure~\ref{Fig:comparison of recursive} (middle) illustrates the non-recursive divide-and-conquer strategy employed in \textsc{FLASH Viterbi}. It can be observed that the execution order, denoted by green solid arrows, strictly adheres to the subtask generation hierarchy, without incurring recursive memory overhead.

\setlength{\textfloatsep}{6pt}   
\setlength{\intextsep}{6pt}      
\captionsetup{aboveskip=2pt, belowskip=2pt} 
\SetAlCapSkip{2pt}

\begin{algorithm}[tb]
\footnotesize
\SetInd{0.3em}{0.6em}
\DontPrintSemicolon
\LinesNumbered
\SetCommentSty{mycommentstyle}
\SetKwComment{Comment}{$\triangleright$\ }{}
\caption{Non-Recursive Divide-and-Conquer}
\label{Alg:Non-recursive divide-and-conquer}
\KwIn{Sequence length $T$, parallelism degree $P$, task queue $\mathcal{Q}_{\text{Task}}$}
\KwOut{Decoded sequence}

$TotalTask \leftarrow T - P + 1$, $TaskCount \leftarrow 0$ \Comment*[r]{Init counter} 
Decoding($\mathcal{Q}_{\text{Task}}, 0, T-1, P$) \Comment*[r]{Decode initial subtask}
InitialEnqueue($\mathcal{Q}_{\text{Task}}, 0, T-1, P$) \Comment*[r]{Enqueue $P$ subtasks}
\For{$i \leftarrow 1$ \KwTo $P$ \textbf{in parallel}}{
    \While{\textbf{true}}{
        \While{$TaskCount < TotalTask$ \textbf{and} $\mathcal{Q}_{\text{Task}} = \emptyset$}{
            WaitForWakeUp() \Comment*[r]{No task, wait}
        }
        \If{Dequeue($\mathcal{Q}_{\text{Task}}(m,n)$) $= \emptyset$}{
            ProcessEnd() \Comment*[r]{All tasks done}
        }
        $t_{\text{mid}} \leftarrow \lfloor (m + n)/2 \rfloor$ \;
        $TaskCount \leftarrow TaskCount + 1$ \;
        Decoding($m, n$) \Comment*[r]{Subtask decoding}
        \If{$(n - m) > 2$}{
            Enqueue($\mathcal{Q}_{\text{Task}}, m, t_{\text{mid}}$) \;
            Enqueue($\mathcal{Q}_{\text{Task}}, t_{\text{mid}}+1, n$) \;
        }
        \ElseIf{$(n - m) = 2$}{
            Enqueue($\mathcal{Q}_{\text{Task}}, m, t_{\text{mid}}$) \;
        }
        WakeUpAllThreads() \Comment*[r]{Notify other threads}
    }
}
\end{algorithm}

\subsubsection{Optimization for Parallel Decoding}

To enhance performance, further optimizations are introduced in the parallel implementation of \textsc{FLASH Viterbi}. When the parallelism degree ($P$) is high, using a binary bisection-based divide-and-conquer strategy—starting from the midpoint—may fail to generate a sufficient number of subtasks during early decomposition, resulting in idle threads and underutilized hardware resources. To address this, our implementation performs a $P$-way partition during the initial subtask, dividing the observation sequence into $P$ equal segments generating $P$ subtasks. Subsequent subtasks decoding proceeds via binary bisection. As illustrated in Figure~\ref{Fig:comparison of recursive} (bottom), for a parallelism degree of $P=4$, our approach bypasses the yellow-colored bisection phase and directly partitions the observation sequence into four (purple) subtasks after initial subtask, thereby enabling immediate utilization of all four threads. While this initial $P$-way partitioning introduces a marginal increase in memory overhead compared to pure bisection, the overhead is negligible relative to the memory footprint of later multithreaded processing and does not impact the overall memory usage. This strategy ensures full thread utilization from the outset while maintaining low per-subtask memory overhead.

Algorithm~\ref{Alg:Non-recursive divide-and-conquer} outlines the parallel decoding workflow of \textsc{FLASH Viterbi}. After decoding the initial subtask, $P$ subtasks are enqueued into the task queue and decoded in parallel by $P$ threads. Each thread enqueues new subtasks via binary bisection. Notably, except for the initial task, each decoding subtask outputs the optimal state at its segment’s midpoint $t_{\text{mid}}$, as indicated by red numbers in Figure~\ref{Fig:comparison of recursive}. When the segment length reduces to 3, the right subtask shares the same midpoint as its parent, obviating additional enqueuing.

\subsubsection{Parallelism Analysis under Non-recursive Divide-and-Conquer}

In Figure~\ref{Fig:comparison of recursive}, the width of each rectangle denotes the sequence length of a subtask, while the area approximates its decoding time. Using the full decoding time of a length-$T=16$ sequence (i.e., the blue rectangle) as the decoding cycles, we quantify the temporal consumption across different strategies. The sequential divide-and-conquer Viterbi requires $\log T$ decoding cycles (e.g., 4 cycles for $T=16$) to complete execution. In contrast, the multithreaded variant with a parallelism degree $P=4$ completes the decoding process in only $1 + \frac{\log T - \log P}{P}$ cycles (e.g., 1.5 cycles for $T=16$), illustrating significant parallel acceleration.

Through this optimization, \textsc{FLASH Viterbi} eliminates the inherent limitations of recursive decomposition while establishing the foundation for parallel execution of subtasks.

\subsection{Pruning and Parallelization Strategy}
\label{sec_tech2}
In this section, we detail our pruning and parallelization strategy designed to optimize the computational process of subtask decoding. 
Our goal is to reduce redundant computations and eliminate dependencies among intra-layer subtasks, thereby enabling parallel execution. We also analyze time and space complexity and provide the proof of correctness of the strategy.

\subsubsection{Subtask Decoding and Dependency-Induced Challenges}

To illustrate the necessity of the pruning and parallelization strategy, we begin by describing the subtask decoding procedure in \textsc{FLASH Viterbi} and analyze the computational dependencies across subtasks. Algorithm~\ref{Alg:FLASH Viterb Decoding Process} presents the subtask decoding process of \textsc{FLASH Viterbi} under binary bisection before applying the pruning and parallelization strategy. Figure~\ref{Fig:FLASH Viterbi decoding process} illustrates this process for a sequence length of $T = 8$ and a state space size of $K = 4$.
As shown in Step~2 of Figure~\ref{Fig:FLASH Viterbi decoding process}, before the backtracking phase, the standard Viterbi algorithm retains all state transitions along each candidate path throughout dynamic programming (depicted by dashed lines). In contrast, \textsc{FLASH Viterbi} retains only the transitions between the division point timestep $t_{\text{mid}}$ (i.e., the midpoint for binary bisection) and the terminal timestep (depicted by red solid lines). This selective retention reduces the space complexity from $\mathcal{O}(KT)$ to $\mathcal{O}(K)$.
This optimization is realized through iterative updates of three key arrays:

\begin{itemize}[leftmargin=*,itemsep=2pt,topsep=3pt,parsep=0pt]
\item \textit{OptProb}: Stores the maximum path probability for each state at the current timestep. 
\item \textit{PreState}: Stores the state index at the previous timestep for each maximum path. 
\item \textit{MidState}: Stores the state index at the division point timesteps of each maximum path, as shown by the red lines in Figure~\ref{Fig:FLASH Viterbi decoding process}.
\end{itemize}

As shown in Step~1 of Figure~\ref{Fig:FLASH Viterbi decoding process} (Dynamic Programming), at timestep $t=5$, the \textit{MidState} array is updated by combining the \textit{PreState} information at $t=5$ (brown dashed lines) with the \textit{MidState} values at $t=4$ (purple solid lines), resulting in the updated \textit{MidState} at $t=5$ (red solid lines). 
This iterative update continues throughout the dynamic programming phase, capturing only the transitions at division points and the terminal timestep. Notably, the updates to \textit{PreState} and \textit{MidState} are initiated only after the division point timestep (i.e., $t_{\text{mid}}+1$ under binary bisection).
During backtracking, the optimal path (depicted by the green dashed line) is identified by comparing the \textit{OptProb} values across all $K$ states. Since \textsc{FLASH Viterbi} retains state transitions only at division points and the terminal timestep, it can directly infer optimal states at these specific timesteps (i.e., red-highlighted states in the decoded sequence in Figure~\ref{Fig:FLASH Viterbi decoding process}). The remaining optimal states are recovered through decoding of subsequent subtasks.

During the divide-and-conquer phase (Step 3 in Figure~\ref{Fig:FLASH Viterbi decoding process}), the decoding task $(0,7)$ is partitioned into two subtasks: $(0,3)$ and $(4,7)$. Although these subtasks are structurally independent, the initialization of \textit{OptProb} at timestep 4 in subtask $(4,7)$ relies on the \textit{OptProb} values at timestep 3 from subtask $(0,3)$. This cross-subtask dependency—evident in line 8 of Algorithm~\ref{Alg:FLASH Viterb Decoding Process}—inhibits their parallel execution.

\begin{algorithm}[tb]
\footnotesize
\SetInd{0.3em}{0.6em}
\DontPrintSemicolon
\LinesNumbered
\SetKwComment{Comment}{$\triangleright$\ }{}
\SetCommentSty{mycommentstyle}
\caption{Subtask Decoding (Before Pruning)}
\label{Alg:FLASH Viterb Decoding Process}
\KwIn{Initial state $\boldsymbol{\pi}$, transition matrix $\mathcal{A}$, emission matrix $\mathcal{B}$, sequence $X$, length $T$, start index $m$, end index $n$}
\KwOut{Optimal states $q_{t_{\text{mid}}}^*$ and $q_{T-1}^*$}

$t_{\text{mid}} \leftarrow \lfloor (m + n)/2 \rfloor$ \Comment*[r]{Midpoint of segment}
\textcolor{mygreen}{\ttfamily{/* Initialization phase */}}

\If{$m = 0$}{
    \For{$i \leftarrow 1$ \KwTo $K$}{
        $\textit{OptProb}[i] \leftarrow \log(\pi_i) + \log(\mathcal{B}_{i,x_0})$ \Comment*[r]{Base case init}
    }
}
\Else{
    \For{$i \leftarrow 1$ \KwTo $K$}{
        $\textit{OptProb}[i] \leftarrow \max_k (\textit{OptProb}[k] + \log(\mathcal{A}_{k,i}) + \log(\mathcal{B}_{i,x_m}))$ \;
    }
}

\For{$i \leftarrow 1$ \KwTo $K$}{
    $\textit{PreState}[i] \leftarrow -1$, $\textit{MidState}[i] \leftarrow -1$ \Comment*[r]{Reset trace}
}
\textcolor{mygreen}{\ttfamily{/* Dynamic programming loop */}}

\For{$t \leftarrow m + 1$ \KwTo $n$}{
    \For{$i \leftarrow 1$ \KwTo $K$}{
        $\textit{OptProb}[i] \leftarrow \max_k (\textit{OptProb}[k] + \log(\mathcal{A}_{k,i}) + \log(\mathcal{B}_{i,x_t}))$ \;
        $\textit{PreState}[i] \leftarrow \arg\max_k (\textit{OptProb}[k] + \log(\mathcal{A}_{k,i}) + \log(\mathcal{B}_{i,x_t}))$ \;        
        \If{$t = t_{\text{mid}} + 1$}{
            $\textit{MidState}[i] \leftarrow \textit{PreState}[i]$ \;
        }
        \ElseIf{$t > t_{\text{mid}} + 1$}{
            $\textit{MidState}[i] \leftarrow \textit{MidState}[\textit{PreState}[i]]$ \;
        }
    }
}
\textcolor{mygreen}{\ttfamily{/* Backtracking phase */}}

\If{$t_{\text{mid}} = \lfloor (T - 1)/2 \rfloor$}{
    $q_{T - 1}^* \leftarrow \arg\max_k (\textit{OptProb}[k])$ \;
    $q_{t_{\text{mid}}}^* \leftarrow \textit{MidState}[q_{T - 1}^*]$ \;
}
\Else{
    $q_{t_{\text{mid}}}^* \leftarrow \textit{MidState}[q_n^*]$ \;
}
\end{algorithm}

\begin{figure}[tb]
		\centerline{\includegraphics[width=1\linewidth]{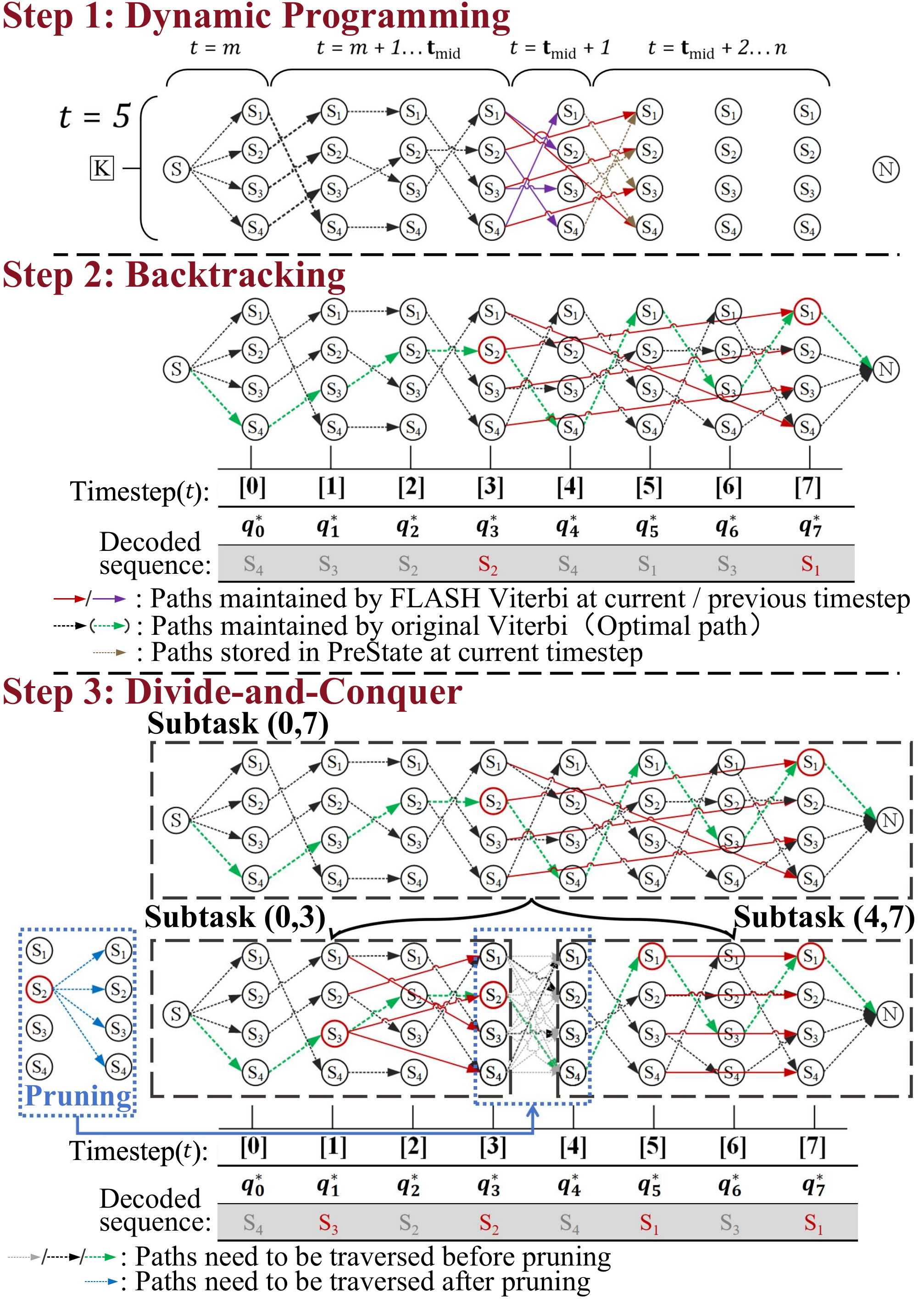}}
		\caption{\textsc{FLASH Viterbi} subtask decoding process.}
		\label{Fig:FLASH Viterbi decoding process}
	\end{figure}

\subsubsection{Pruning for Parallel Decoding}
To eliminate cross-subtask dependencies and enable parallel decoding, we propose the following pruning strategy: when solving a subtask starting at the division point timestep of the previous divided layer, only the transition paths originating from the optimal state at that timestep are retained, and their corresponding \textit{OptProb} values are re-initialized. This approach is illustrated in Step 3 of Figure~\ref{Fig:FLASH Viterbi decoding process}, highlighted by the two blue dashed boxes.

In the right dashed box, the dashed lines depict all transition paths that need to be traversed at timestep $t{=}3$ prior to pruning. Among them, black and green lines indicate maximum-probability transitions for each state. In contrast, the blue dashed lines in the left box depict the pruned transition paths, where only transitions leading to the optimal state \( S_2 \) at \( t{=}3 \) are retained. The number of transition computations eliminated through pruning is given by $(K{-}1)K\sum_{i=1}^{\log T - 1} \frac{T}{2^i}$, which matches the overall time complexity of the full decoding process, i.e., $\mathcal{O}(K^2 T \log T)$. This highlights the pruning strategy's effectiveness in eliminating redundant computations.
The rationale behind this strategy is that subsequent subtasks primarily serve to reconstruct the optimal path. Thus, it is sufficient to preserve only those transitions along the optimal path from the original decoding task and ensure their corresponding probabilities remain maximal within each subtask. Transitions associated with non-optimal candidate paths do not need to maintain consistency with the original global decoding.
Taking Step~3 of Figure~\ref{Fig:FLASH Viterbi decoding process} as an example, the maximum-probability paths for each state in the right dashed box exhibit probabilities no lower than those of the retained paths in the left box after pruning. Notably, the optimal path—originally the highest-probability path in the right box—remains optimal in the pruned subset shown in the left box. Therefore, this pruning strategy does not alter the final decoding result. To simplify computations, we set the accumulated path probability of the optimal state at the initial timestep to 1 (e.g., state \( S_2 \) at timestep \( t=3 \) in Step~3 of Figure~\ref{Fig:FLASH Viterbi decoding process}). 
By applying logarithmic transformation, the dependency on the previous timestep’s \textit{OptProb} can be eliminated in the path probability equation, enabling parallel subtask execution. 
Accordingly, line~8 of Algorithm~\ref{Alg:FLASH Viterb Decoding Process} is modified to: 
\[
\textit{OptProb}[i] = \log(\mathcal{A}_{q_{m-1}^*,i}) + \log(\mathcal{B}_{i,x_m}).
\]
With a parallelism degree of \( P \), this optimization reduces the time complexity from \( \mathcal{O}(K^2T\log{T}) \) to \( \mathcal{O}\left(K^2T\frac{(\log{T} - \log{P})}{P}\right) \), 
and the space complexity to \( \mathcal{O}(PK) \).

Additionally, the initial probability computation for a subtask starting at a division point is structurally analogous to that of a subtask starting at timestep 0. In both cases, decoding is initiated from a single state (e.g., state \( S_1 \) for subtask \((0,3)\) and state \( S_2 \) for subtask \((4,7)\) in Figure~\ref{Fig:FLASH Viterbi decoding process}). This structural similarity enables the use of a unified hardware architecture for all subtasks during edge-side deployment, thereby improving hardware efficiency and resource utilization.

\subsubsection{Correctness of Pruning and Parallelization Strategy}
The correctness of the divide-and-conquer formulation for the Viterbi algorithm has been proven in~\cite{SIEVE}. This section therefore focuses on proving the correctness of the proposed pruning and parallelization strategy. We first formally define the valid path set and the optimal path as follows.

\noindent
\begin{definition}[Valid path set]
Given a divide-and-conquer subtask $(m, n)$ with a state space of size $K$ and $m \neq 0$, the valid path set $\mathcal{Y}$ is defined as $\mathcal{Y} = \{Q_1, Q_2, \dots, Q_y\}$, representing all feasible paths from timestep $m{-}1$ to $n$, with corresponding path probabilities $\mathcal{P} = \{p_1, p_2, \dots, p_y\}$. The cardinality $y$ satisfies $1 \leq y \leq K^{n - m + 2}$, depending on the sparsity of the transition matrix $\mathcal{A}$.
\end{definition}

\begin{definition}[Optimal path]
The optimal path $Q^*$ under the divide-and-conquer strategy is defined as $Q^* = [q_{m-1}^*, q_m^*, \dots, q_n^*]$, which maximizes the path probability, i.e., $p^* = \max(\mathcal{P})$.
\end{definition}

For clarity, we define the following two variants of the pruning strategy used in the proofs.

\begin{definition}[Retained-pruning variant]
The retained-pruning variant restricts the path set to those sharing the same $(m{-}1)$-th state $q_{m-1}$ as the optimal path $Q^*$. It defines a subset $\mathcal{Y}' = \{Q \in \mathcal{Y} \mid q_{m-1} = q_{m-1}^*\}$ with path probabilities $\mathcal{P}' = \{p_1', p_2', \dots\}$. The optimal path 
is $Q' = [q_{m-1}', q_m', \dots, q_n']$, which maximizes the path probability, i.e., $p' = \max(\mathcal{P}')$.
\end{definition}

\begin{definition}[Complete-pruning variant]
The complete-pruning variant extends the retained-pruning variant by 
assuming $p_{m-1}^* = 1$. Decoding is performed using $\mathcal{Y}'' = \mathcal{Y}'$ with path probabilities $\mathcal{P}'' = \{p_1'', p_2'', \dots\}$. The optimal path 
is $Q'' = [q_{m-1}'', q_m'', \dots, q_n'']$, which maximizes the path probability, i.e., $p'' = \max(\mathcal{P}'')$.
\end{definition}

\begin{theorem}
\label{theo_1}
The optimal path $Q^*$ belongs to the retained-pruning path set, i.e., $Q^* \in \mathcal{Y}'$.
\end{theorem}

\begin{proof}
From Definition~{3}, we have $\mathcal{Y}' = \{Q \in \mathcal{Y} \mid q_{m-1} = q_{m-1}^*\}$.  
Since $Q^* = [q_{m-1}^*, q_m^*, \dots, q_n^*]$ satisfies $q_{m-1} = q_{m-1}^*$, it follows that $Q^* \in \mathcal{Y}'$. Moreover, $p^* \in \mathcal{P}'$.
\end{proof}

\begin{theorem}
\label{theo_2}
The optimal path from the retained-pruning variant equals the one from divide-and-conquer, i.e., $Q' = Q^*$.
\end{theorem}

\begin{proof}
Since $\mathcal{Y}' \subseteq \mathcal{Y}$ and $p' = \max(\mathcal{P}')$, and from Theorem~{1}, we have $p^* \in \mathcal{P}'$, then $p' = p^*$. Thus, the optimal paths are the same: $Q' = Q^*$.
\end{proof}

\begin{theorem}
\label{theo_3}
The optimal path from complete-pruning variant equals the one from divide-and-conquer, i.e., $Q'' = Q^*$.
\end{theorem}

\begin{proof}
From Definitions~{3} and~{4}, the update rule in retained-pruning at time $m$ is $\textit{OptProb}'[i] = \log(\mathcal{A}_{q_{m-1}^*, i}) + \log(\mathcal{B}_{i, x_m}) + \log(p_{m-1}^*)$,
while in complete-pruning it becomes $\textit{OptProb}''[i] = \log(\mathcal{A}_{q_{m-1}^*, i}) + \log(\mathcal{B}_{i, x_m})$.
Since $\log(p_{m-1}^*)$ is constant w.r.t. $i$, it does not affect the $\arg\max$ operation. Therefore, $\textit{MidState}'' = \textit{MidState}'$. From Algorithm \ref{Alg:FLASH Viterb Decoding Process}, the optimal states are maintained based on \textit{MidState}. Hence $Q'' = Q'$, and by Theorem~{2}, $Q'' = Q^*$.
\end{proof}

These theorems together confirm that the pruning and parallelization strategy preserves the correctness of the divide-and-conquer approach, yielding the same optimal path.
Overall, the pruning and parallelization strategy enhances HMM decoding in three key aspects:  
1) Reducing redundant computations without sacrificing memory efficiency;  
2) Removing inter-subtask dependencies to enable parallel decoding;  
3) Unifying subtask computational structures, enhancing resource utilization.

\subsection{\textsc{FLASH Viterbi} with Beam Search}
\label{sec_tech3}

To address memory constraints in standard \textsc{FLASH Viterbi} and further enhance its adaptability, we propose \textsc{FLASH-BS Viterbi}—a dynamic beam search variant. This section first introduces the key idea, then details an efficient heap-based mechanism to reduce computational overhead, and finally analyzes the resulting complexity.

\subsubsection{Static vs. Dynamic Beam Search}

Prior works have integrated beam search with divide-and-conquer Viterbi algorithm. However, these methods typically adopt static beam selection, wherein the path probabilities for all candidate states are computed and stored before selecting the top-$B$ candidates.
This static beam search suffers from limited memory optimization because it still requires storing the state information of the remaining \( K-B \) states before discarding them. 
In contrast, \textsc{FLASH-BS Viterbi} employs dynamic beam search, which incrementally maintains only the top-$B$ states during the process of calculating path probabilities. In this way, \textsc{FLASH-BS Viterbi} eliminates the need to store full intermediate results for the remaining $K-B$ states.

\subsubsection{Efficient Heap-Based Candidate Maintenance}

To mitigate the overhead caused by frequent comparisons and selective replacements in dynamic candidate maintenance, we introduce an efficient data structure based on two min-heaps for optimal state management, as illustrated in Figure~\ref{Fig:Efficient Heap-Based Storage Structure}.

    \begin{figure}
		\centerline{\includegraphics[width=1\linewidth]{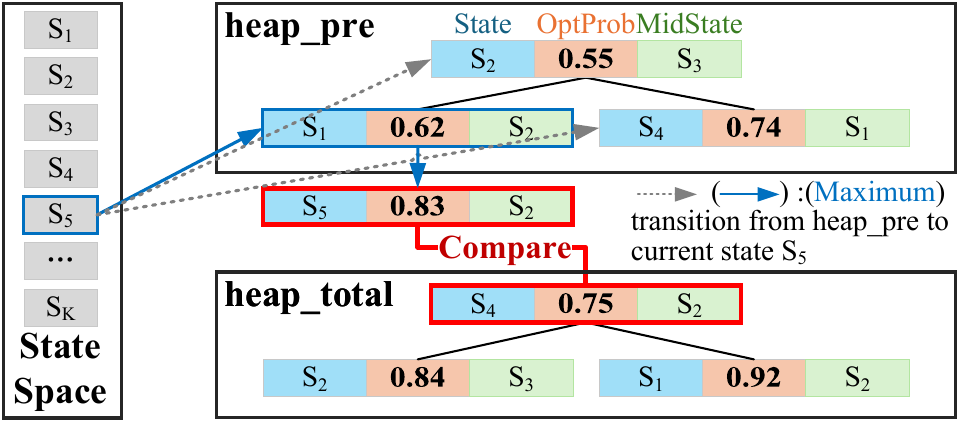}}
		\caption{Efficient data structure for dynamic beam search.}
	\label{Fig:Efficient Heap-Based Storage Structure}
	\end{figure}
The algorithm employs two min-heaps: (i) $\textit{heap\_pre}$, which retains the top \( B \) candidate states from the previous timestep, and (ii) $\textit{heap\_total}$, which incrementally updates the top \( B \) candidates for the current timestep.
Each heap maintains three core components: (i) $\textit{State}$, which records the indices of the top-$B$ candidate states; (ii) $\textit{OptProb}$, which stores their corresponding path probabilities and serves as the heap’s sorting key; (iii) $\textit{MidState}$, which tracks the state indices at division points timestep along the maximum paths, consistent with the \textsc{FLASH Viterbi} definition.
To avoid data copying overhead, two physically separate heap buffers are allocated and managed using a double-buffering scheme. The roles of $\textit{heap\_pre}$ and $\textit{heap\_total}$ alternate across timesteps, enabling efficient in-place updates and streamlined memory access.

The heap update procedure is depicted in Figure~\ref{Fig:Efficient Heap-Based Storage Structure}. During decoding, \textsc{FLASH-BS Viterbi} evaluates only transitions from the \( B \) candidate states in $\textit{heap\_pre}$ to the current state, generating a corresponding heap element for each. The update logic for $\textit{heap\_total}$ proceeds as follows: 1) If fewer than \( B - 1 \) elements are stored, the new element is inserted directly. 2) If exactly \( B - 1 \) elements are present, the new element is inserted and the heap is heapified. 3) If \( B \) elements are already stored, the new element is compared against the current heap root (i.e., the candidate with the lowest \textit{OptProb}). It is inserted only if its \textit{OptProb} is higher, thereby maintaining the top-\( B \) candidates.

\subsubsection{Analysis of Efficiency and Adaptability}
Leveraging the efficient memory design in dynamic beam search, \textsc{FLASH-BS Viterbi} significantly reduces memory consumption and enhancing adaptability. Given a parallelism degree \(P\) and beam width \(B\), the time complexity is reduced to \(\mathcal{O}\left(BKT \cdot \frac{(\log{T} - \log{P})}{P}\right)\), as each state only computed transitions from the top \(B\) candidate paths at the previous timestep. Correspondingly, the space complexity is reduced to 
$\mathcal{O}(PB)$, since only \(B\) paths per parallel unit need to be maintained. The two tunable parameters—parallelism degree \(P\) and beam width \(B\)—demonstrate the multidimensional adaptability of \textsc{FLASH-BS Viterbi}, enabling fine-grained trade-offs between decoding latency and memory consumption.
\section{Hardware Implementation\label{sec_tech4}}

In this section, we present the hardware-accelerated implementation of \textsc{FLASH-BS Viterbi}. 
Compared with existing space-efficient Viterbi algorithms, \textsc{FLASH Viterbi} and \textsc{FLASH-BS Viterbi} exhibit hardware-friendly properties, including non-recursive structures, the elimination of BFS traversal, and double-buffered memory schemes. These features enable efficient parallelization and low-overhead deployment on edge devices. The following describes the accelerator architecture and its FPGA implementation.

\begin{figure}[t!]
    \centering
    \includegraphics[width=1\columnwidth]{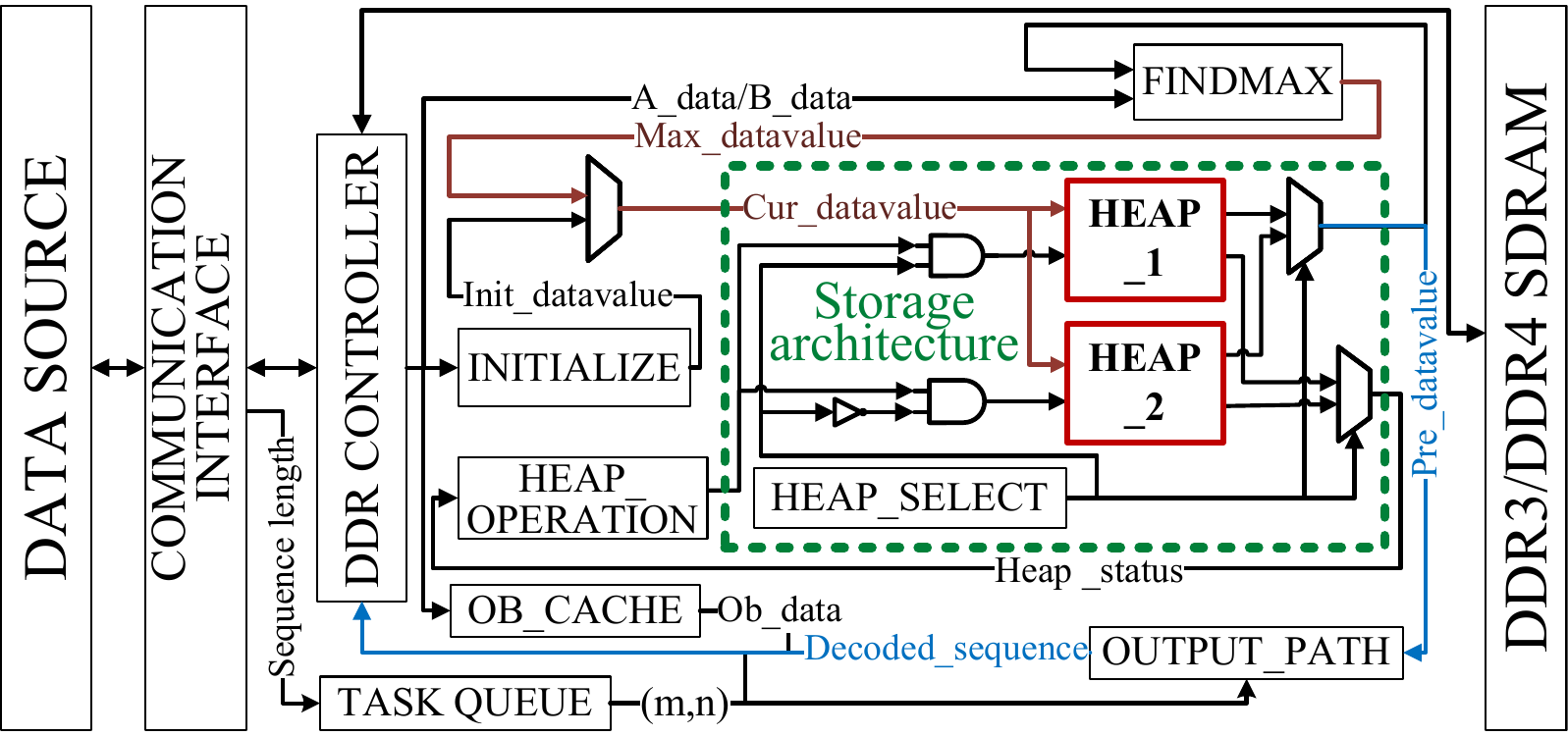}
    \caption{Architecture of the \textsc{FLASH-BS Viterbi} accelerator.}
    \label{fig:FLASH_Decoding_unit}
\end{figure}

\subsection{FPGA-Based Accelerator Architecture}

Figure~\ref{fig:FLASH_Decoding_unit} shows the \textsc{FLASH-BS Viterbi} accelerator architecture. After loading HMM data into DDR memory, decoding is initiated on the FPGA.
A FIFO-based \textit{TASK QUEUE} generates subtask tuples based on the sequence length and sequentially dispatches them to the \textsc{DDR CONTROLLER} to fetch the corresponding observation segments from DDR memory. Each segment is cached in the \textit{OB\_CACHE} in decoding order prior to execution.
During subtask initialization, the relevant data is fetched from DDR memory to perform the initial path probability computation (i.e., \textit{Init\_datavalue} in Figure~\ref{fig:FLASH_Decoding_unit}).
For subsequent timesteps, the \textsc{FINDMAX} unit performs dynamic programming using HMM data from DDR memory and intermediate results (\textit{Pre\_datavalue}) stored in the storage units. The updated results are written back, as shown by the red paths in Figure~\ref{fig:FLASH_Decoding_unit}. In the final backtracking phase, the storage units write the traced state sequence to the \textit{OUTPUT\_PATH}. After all subtasks are processed, the complete \textit{Decoded\_sequence} is written to DDR memory, finalizing the decoding process (blue paths in Figure~\ref{fig:FLASH_Decoding_unit}).

\begin{table*}[htbp]
    \centering
    \setlength{\tabcolsep}{2pt}
    \renewcommand{\arraystretch}{1.05}
    \caption{Overall comparison of time and space for \textsc{FLASH} variants against baseline Viterbi algorithms, for both C and Python implementations. Speedups and memory ratios of our proposed algorithms over other algorithms are also reported.
    }
    \label{tab:overall-comparison}
    \vspace{1mm}
    \resizebox{\textwidth}{!}{%
    \begin{tabular}{|c|c|c|ccc|ccc|ccc|ccc|cc|}
    \hline
    \multirow{2}{*}{\textbf{Class}} & \multirow{2}{*}{\textbf{Algorithm}} & \multirow{2}{*}{\textbf{Impl.}}
    & \multicolumn{3}{c|}{\textbf{Decoding time (s)}} 
    & \multicolumn{3}{c|}{\textbf{Speedup}} 
    & \multicolumn{3}{c|}{\textbf{Memory usage ($\times 10^3$ B)}} 
    & \multicolumn{3}{c|}{\textbf{Memory ratio}}
    & \multicolumn{2}{c|}{\textbf{Py/C ratio}} \\
    \cline{4-17}
     &  &  & Seq. & Par.7 & Par.16 & Seq. & Par.7 & Par.16 & Seq. & Par.7 & Par.16 & Seq. & Par.7 & Par.16 & Time & Mem \\
    \hline
    \multirow{5}{*}{\shortstack[l]{Beam search\\integration\\(B=128)}}
      & \multirow{2}{*}{\centering SIEVE-BS} & C  & 49.4 & -- & -- & 3.5 & 14.3 & 18.3 & 2899.4 & -- & -- & 901.5 & 132.8 & 58.2 & -- & -- \\
      &                                     & Py & 208.5 & -- & -- & 14.7 & 60.2 & 77.4 & 22928.1 & -- & -- & 7129.4 & 1049.8 & 460.6 & 4.2 & 7.9 \\
      \cline{2-17}
      & \multirow{2}{*}{\centering SIEVE-BS-Mp} & C  & 38.0 & -- & -- & 2.7 & 11.0 & 14.1 & 2503.3 & -- & -- & 778.4 & 114.6 & 50.3 & -- & -- \\
      &                                        & Py & 183.1 & -- & -- & 12.9 & 52.9 & 68.0 & 20808.8 & -- & -- & 6470.4 & 952.8 & 418.0 & 4.8 & 8.3 \\
      \cline{2-17}
      & \textsc{FLASH-BS} & C  & 14.2 & 3.5 & 2.7 & 1.0 & 1.0 & 1.0 & 3.2 & 21.8 & 49.8 & 1.0 & 1.0 & 1.0 & -- & -- \\
    \hline
    \multirow{7}{*}{\shortstack[l]{Without\\beam search}}
      & Vanilla & C  & 104.4 & -- & -- & 0.3 & 1.1 & 1.5 & 8120.3 & -- & -- & 127.8 & 18.3 & 8.0 & -- & -- \\
      & Viterbi & Py & 151.7 & -- & -- & 0.4 & 1.6 & 2.1 & 32549.3 & -- & -- & 512.1 & 73.3 & 32.1 & 1.5 & 4.0 \\
      \cline{2-17}
      & Checkpoint & C  & 216.0 & -- & -- & 0.6 & 2.3 & 3.0 & 824.8 & -- & -- & 13.0 & 1.9 & 0.8 & -- & -- \\
      & Viterbi    & Py & 374.5 & -- & -- & 1.0 & 4.0 & 5.2 & 3271.8 & -- & -- & 51.5 & 7.4 & 3.2 & 1.7 & 4.0 \\
      \cline{2-17}
      & \multirow{2}{*}{\centering SIEVE-Mp} & C  & 672.6 & -- & -- & 1.7 & 7.2 & 9.4 & 775.0 & -- & -- & 12.2 & 1.7 & 0.8 & -- & -- \\
      &                                     & Py & 1679.8 & -- & -- & 4.4 & 17.9 & 23.4 & 127554.1 & -- & -- & 2006.8 & 287.1 & 125.6 & 2.5 & 164.6 \\
      \cline{2-17}
      & \textsc{FLASH} & C  & 385.9 & 93.6 & 71.7 & 1.0 & 1.0 & 1.0 & 63.6 & 444.2 & 1015.3 & 1.0 & 1.0 & 1.0 & -- & -- \\
    \hline
    \end{tabular}%
    }
    \vspace{-4mm}
\end{table*}

\subsection{Memory Optimization}
To eliminate the substantial time and resource overhead caused by data transfer operations, the storage architecture employs a double-buffered memory scheme. 
In the \textsc{FLASH-BS Viterbi} accelerator, two BRAM-based min-heap structures, \textit{HEAP\_1} and \textit{HEAP\_2}, dynamically alternate roles between \textit{Heap\_total} and \textit{Heap\_pre}, as described in Section~\ref{sec_tech3}. 
These heap structures perform search and update operations based on instructions issued by the \textsc{HEAP\_OPERATION} module.

\subsection{Parallelization and Pipelining Optimization}

To accelerate decoding, our design applies optimizations to both data access and value computation. Specifically, multiple data groups are fetched from DDR memory in a single clock cycle and concurrently processed by the \textsc{FINDMAX} module.
Additionally, the design employs pipelining to overlap memory access with computation in the \textsc{FINDMAX} module, thereby hiding memory latency and improving overall throughput.
While beam search introduces irregular memory access patterns that may limit data parallelism, it significantly reduces on-chip memory usage, enabling more efficient deployment on resource-constrained hardware.

\section{Experiments}
In this section, we present the performance of the \textsc{FLASH Viterbi} and \textsc{FLASH-BS Viterbi} algorithms under various HMM scenarios, compared with several baseline methods.

\subsection{Experimental Setting}

\noindent
\textbf{Data.} 
We generate various Erdős–Rényi transition graphs to simulate different HMM scenarios. 
Specifically, we vary the edge probability \( p \) (i.e., the likelihood of transition between states), the state space size \( K \), the observation sequence length \( T \), and assess each algorithm’s decoding time and memory usage under different conditions. 
Additionally, to evaluate the impact of beam width on memory usage, decoding time, and recognition accuracy, we construct a forced-alignment dataset with \( K \) = 3965 and \( T \) = 256. The dataset is generated using the HTK toolkit~\cite{HTK} to force-align speech transcriptions~\cite{factor} from the TIMIT corpus~\cite{TIMIT}.

\noindent
\textbf{Platform.} 
All CPU experiments used 2×16-core Xeon 6226R CPUs (2.90\,GHz, 512\,GB RAM), and the FPGA accelerator was implemented in Verilog on a Xilinx XCZU7EV at 200\,MHz.

\noindent
\textbf{Baselines.} 
On the CPU platform, we compare the \textsc{FLASH} variants (\textsc{FLASH Viterbi} and \textsc{FLASH-BS Viterbi}) against five baselines:  
(i) Vanilla Viterbi~\cite{viterbi_algorithm}: the standard Viterbi algorithm;  
(ii) Checkpoint Viterbi~\cite{checkpoint}: stores intermediate decoding states every \( \sqrt{T} \) steps to reduce memory usage; 
(iii) SIEVE-Mp~\cite{SIEVE}: applies a recursive divide-and-conquer strategy to reduce memory usage; 
(iv) SIEVE-BS~\cite{SIEVE_bs}: a SIEVE variant with static beam search; 
(v) SIEVE-BS-Mp~\cite{SIEVE_bs}: a SIEVE-Mp variant with static beam search.
All baselines were originally implemented in Python. 
The \textsc{FLASH} variants were implemented in C as it offers higher efficiency and is better suited for resource-constrained edge deployment. To ensure fair comparison, we re-implemented all baselines in C and using these C versions for comparison in all experiments unless specified.
On the FPGA platform, we compare our accelerator with Reduced-Memory Viterbi, a SOTA beam search–based implementation~\cite{ReducedViterbi}. 

\noindent
\textbf{Parameter settings.} 
To evaluate the impact of individual parameters, we use default settings unless otherwise specified: observation space size \( |O| = 50 \), edge probability \( p = 0.253 \), state space size \( K = 512 \), and sequence length \( T = 512 \). The default settings are chosen with reference to~\cite{SIEVE} and are commonly adopted in evaluating Viterbi algorithms.
In beam width experiments using the forced-alignment dataset, the value of \( p \) is inherently determined by the dataset and is unaffected by our default setting. Emission probabilities are randomized. Beam width \( B \) is initially set to \( K \), with all candidate paths retained by default. For edge probability analysis, \( p \) is varied from 0.05 to 1 in multiplicative steps of 1.5. The state space size \( K \) and sequence length \( T \) are varied from 32 to 2048. The beam width \( B \) is decreased from 1024 to 32.

\subsection{Overall Comparison}

\begin{figure*}[htbp]
    \centering
    \includegraphics[width=0.95\textwidth]{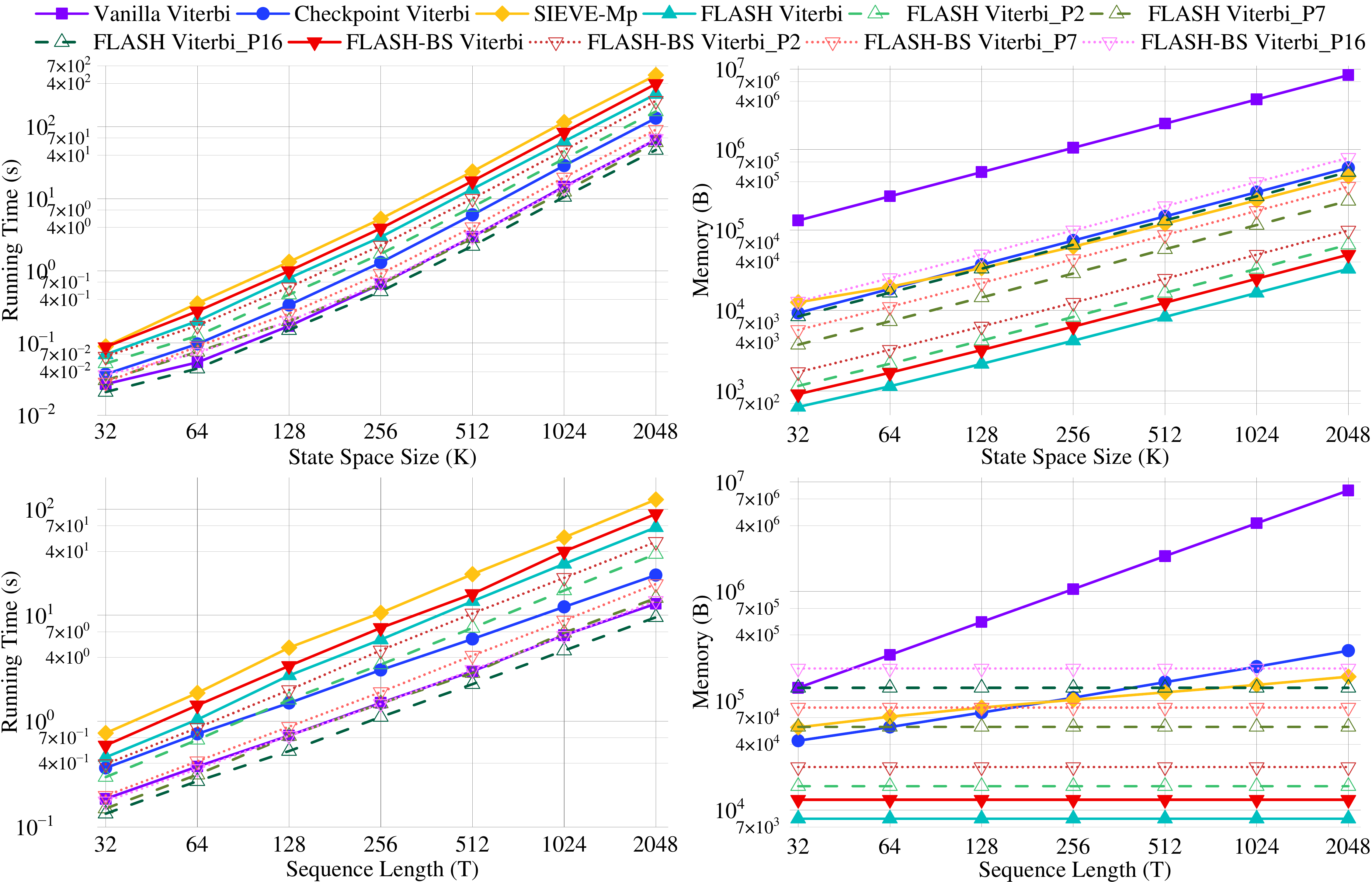}
    \caption{Decoding time and memory usage vs. state space size and sequence length. \textsc{FLASH Viterbi\_P\emph{n}} and \textsc{FLASH-BS Viterbi\_P\emph{n}} represent the algorithms executed under parallelism degree \(n\).}
\label{Fig:Sequence_Length_State_Space_Size}
    \vspace{-2mm}
\end{figure*}

\begin{figure*}[tb]
    \centering
    \includegraphics[width=0.95\textwidth]{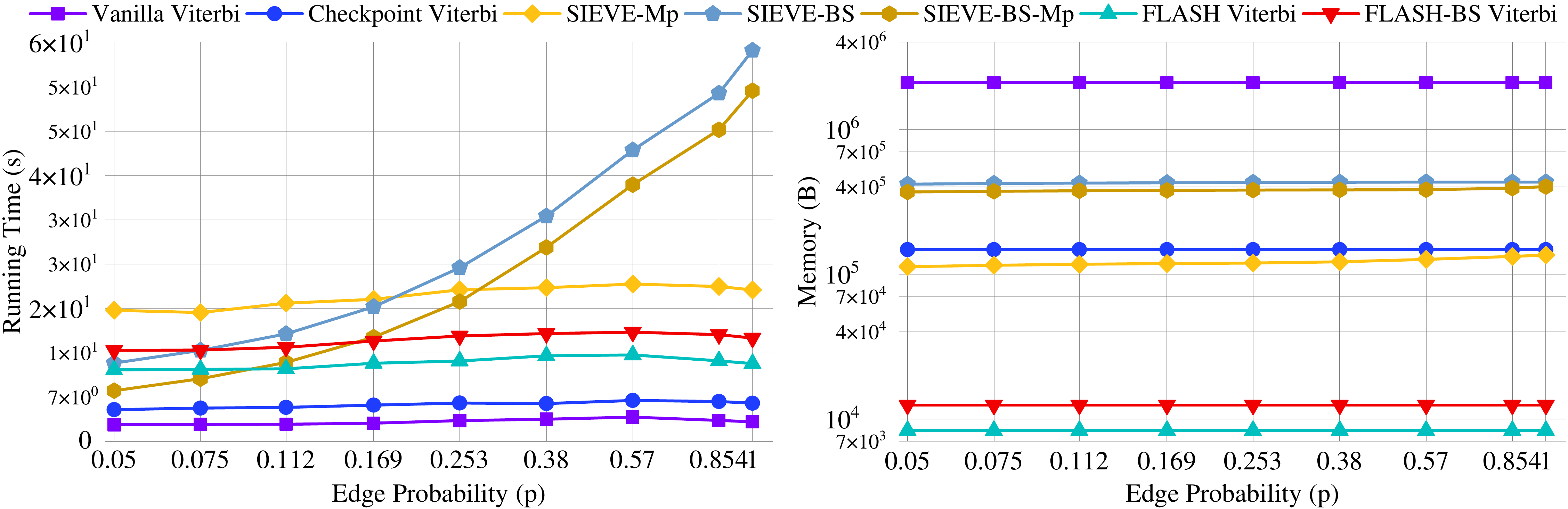}
    \caption{Decoding time and memory usage vs. edge probability.}
    \label{Fig:edge probability}
    \vspace{-4mm}
\end{figure*}

In this section, we evaluate the overall performance of \textsc{FLASH} variants against baseline algorithms in both their original Python and optimized C implementations, using the forced-alignment dataset. 
We report decoding time (in seconds) and memory usage (in bytes). Decoding time is measured from the input of the HMM model to the completion of the optimal path, while memory usage includes all relevant data structures used during decoding. 
We also report the performance gains of the C-optimized baselines over their original Python versions in the \textbf{Py/C ratio} columns.
We evaluate algorithms in two categories: with beam search integration and without beam search. For each category, we compare \textsc{FLASH} variants with existing baselines under different levels of parallelism.

\subsubsection{With Beam Search Integration}

Table~\ref{tab:overall-comparison} reports the results of the comparison. \textsc{FLASH-BS Viterbi} consistently outperforms both SIEVE-BS and SIEVE-BS-Mp in decoding time and memory usage. In sequential execution, it achieves \(\mathbf{3.5\times}\) and \(\mathbf{2.7\times}\) speedups over their C-optimized implementations, while reducing memory consumption by \(\mathbf{901.5\times}\) and \(\mathbf{778.4\times}\), respectively. For the original Python implementations, these advantages become even more significant, with speedups increasing to \(\mathbf{14.7\times}\) and \(\mathbf{12.9\times}\), and memory reductions reaching \(\mathbf{7129.4\times}\) and \(\mathbf{6470.4\times}\).
The decoding speedup primarily arises from three design optimizations: (i) a non-recursive divide-and-conquer strategy that eliminates recursive calls and BFS traversal overhead; (ii) a double-buffered memory scheme that reduces data transfer latency during dynamic programming; and (iii) an integrated pruning-parallelism mechanism that eliminates redundant path computations across subtasks. 
In terms of memory usage, the improvement mainly stems from the efficient memory structure of dynamic beam search. Unlike SIEVE-BS and SIEVE-BS-Mp, which store the entire \(K\) state space, our method maintains only \(B\) candidate states at each timestep, resulting in substantial memory savings.

The advantages of \textsc{FLASH-BS Viterbi} become even more pronounced under parallel execution. At a parallelism level of 16, it achieves speedups of \(\mathbf{18.3\times}\) and \(\mathbf{14.1\times}\) over the C-optimized implementations of SIEVE-BS and SIEVE-BS-Mp, while reducing memory consumption by \(\mathbf{58.2\times}\) and \(\mathbf{50.3\times}\), respectively.
These results indicate that the performance of our algorithm can be further enhanced by leveraging parallelism.

\subsubsection{Without Beam Search Integration}

In the non–beam search category under sequential execution, \textsc{FLASH Viterbi} shows slower speed than the fastest baseline, Vanilla Viterbi. 
However, it achieves substantially lower memory consumption, reducing usage by \(\mathbf{512.1\times}\) in the original Python implementation and by \(\mathbf{127.8\times}\) in the C-optimized version, thereby leading to a markedly improved overall time–space trade-off. 
By tuning the degree of parallelism to balance decoding time and memory usage, \textsc{FLASH Viterbi} achieves a more favorable trade-off. 
As shown in Table~\ref{tab:overall-comparison}, at a parallelism level of 7, it outperforms all baselines in both decoding time and memory usage, demonstrating superior overall resource efficiency.

\subsection{Effect of HMM Conditions}

In this section, we evaluate the scalability and efficiency of various decoding algorithms under diverse HMM settings, focusing on structural scale—determined by sequence length and state space size—and transition density, defined by edge probabilities. We specifically compare the sequential performance of \textsc{FLASH} variants against baseline algorithms.

\begin{figure*}[htbp]
    \centering
    \includegraphics[width=0.95\textwidth]{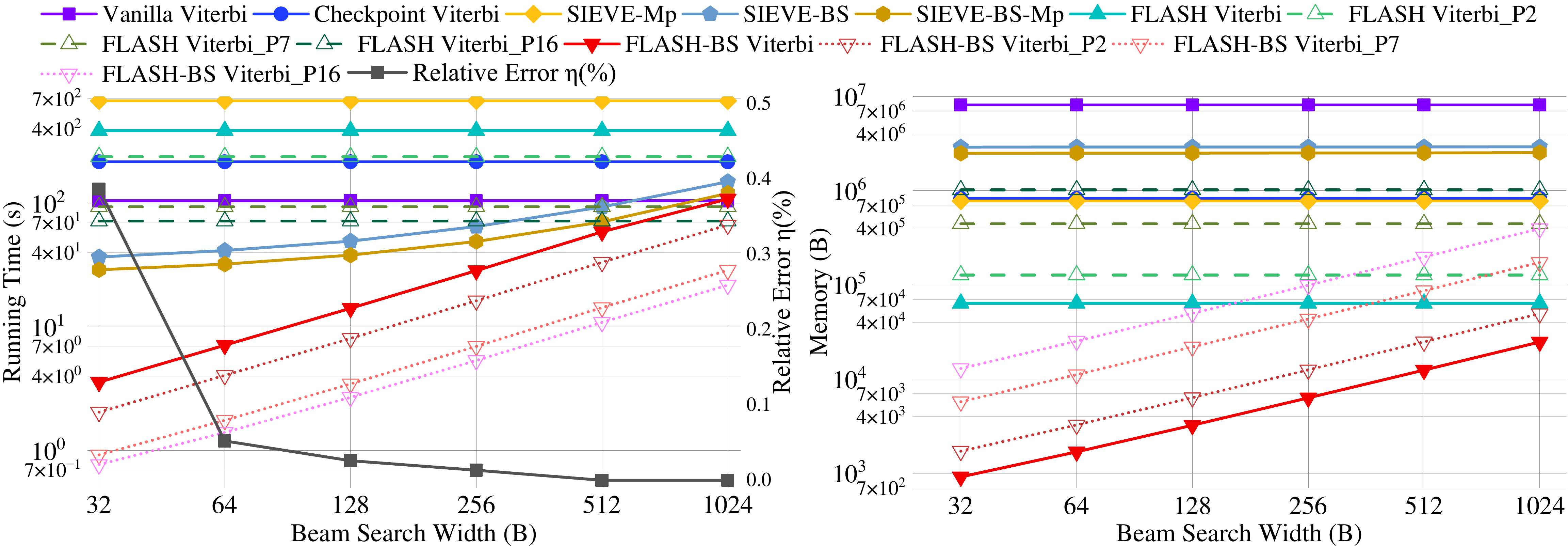}
    \caption{Decoding time, memory usage, and relative error vs. beam search width.  \text{Relative error} reflects decoding accuracy loss.}
    \label{Fig:Beam_width}
    \vspace{-2mm}
\end{figure*}

\begin{figure*}[htbp]
    \centering
    \includegraphics[width=0.95\textwidth]{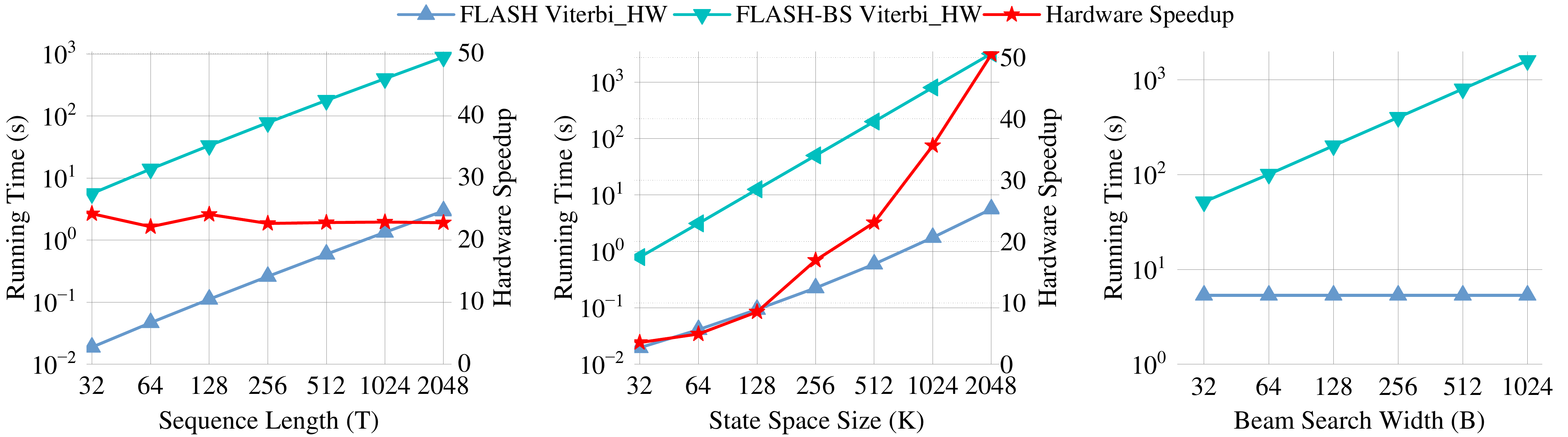}
    \caption{ Decoding time and hardware speedup of \textsc{FLASH Viterbi} and \textsc{FLASH-BS Viterbi} accelerators.}
    \label{Fig:hardware}
    \vspace{-6mm}
\end{figure*}

\subsubsection{Performance Under Varying Sequence Lengths and State Space Sizes}

Figure~\ref{Fig:Sequence_Length_State_Space_Size} presents decoding performance under varying sequence lengths \( T \) and state space sizes \( K \). Although \textsc{FLASH Viterbi} and SIEVE-Mp share the same theoretical complexity of \( \mathcal{O}(K^2 T \log T) \), \textsc{FLASH Viterbi} achieves lower actual runtime due to its non-recursive design and double-buffered memory scheme. 
These optimizations eliminate recursion and BFS overhead in divide-and-conquer steps, while reducing data transfer latency during dynamic programming.
Compared to \textsc{FLASH Viterbi}, \textsc{FLASH-BS Viterbi} shows a slight increase in runtime, primarily due to heap maintenance overhead.
Regarding memory usage, \textsc{FLASH Viterbi} and \textsc{FLASH-BS Viterbi} consistently achieve the lowest footprint across varying \( T \) and \( K \), by storing only the state transitions between the division point and the terminal timestep. This design decouples memory usage from sequence length \( T \). Additionally, the non-recursive formulation and double-buffered memory scheme further reduce memory usage, enabling our algorithms to outperform SIEVE-Mp, which is the most theoretically memory-efficient baseline to date.

\subsubsection{Performance Under Varying Edge Probabilities}

Figure~\ref{Fig:edge probability} presents decoding time and memory usage across varying edge probabilities \( p \). Compared to SIEVE-BS and SIEVE-BS-Mp, \textsc{FLASH} variants
are less sensitive to changes in transition density. This contrast stems from their underlying HMM representations: \textsc{FLASH} variants use a standard state-matrix-based approach, while SIEVE-BS and SIEVE-BS-Mp adopt a token-passing formulation. Although token-passing avoids traversing unreachable states, it incurs additional overhead due to multiple per-state comparisons required to identify the most probable incoming transitions. As a result, SIEVE-BS-Mp achieves faster decoding at extremely sparse settings (\( p \leq 0.075 \)), but its runtime increases rapidly as \( p \) grows. In contrast, \textsc{FLASH} variants maintain stable performance across all values of \( p \), demonstrating robustness to graph sparsity.

\subsection{Impact of Individual Parameters}

In this section, we evaluate the decoding performance of \textsc{FLASH} variants
by adjusting two internal parameters: parallelism degree $n$ and beam width $B$. 

\subsubsection{Impact of Parallelism Degree on Runtime and Memory Usage}

Figure~\ref{Fig:Sequence_Length_State_Space_Size} illustrates the performance of \textsc{FLASH Viterbi} and \textsc{FLASH-BS Viterbi} under different parallelism levels (\(n = 2, 7, 16\)). By tuning the parallelism degree \(n\), \textsc{FLASH Viterbi} outperforms existing baselines in both decoding time and memory usage. For instance, at \(T = 512\) and \(K = 256\)–\(2048\), \textsc{FLASH Viterbi\_P7} achieves lower decoding latency than the fastest baseline (Vanilla Viterbi) while consuming even less memory than the most memory-efficient baseline (SIEVE-Mp). 
This improvement stems from its parallel execution strategy: decoding subtasks are evenly distributed across threads, and each subtask stores the state information at the division point during dynamic programming. By eliminating the backtracking table lookup required by Vanilla Viterbi, this design significantly reduces decoding latency. Although memory consumption scales linearly with the parallelism degree \(n\), the inherently low footprint of \textsc{FLASH Viterbi} allows its parallel variants to remain more memory-efficient than existing baselines. More importantly, unlike fixed-cost baselines, 
\textsc{FLASH Viterbi} can fully utilize memory and compute resources to meet both real-time and deployment requirements.
For example, \textsc{FLASH Viterbi\_P16} achieves the lowest latency among all baselines at all times, with only a slight increase in memory usage. These results highlight the efficiency and adaptability of \textsc{FLASH Viterbi}.

\subsubsection{Trade-offs Introduced by Beam Width in Accuracy and Efficiency}
To assess the impact of beam search on decoding accuracy, we follow~\cite{SIEVE_bs} and report the relative error in log-likelihood, defined as \( \eta = \frac{| \ell_{\text{OPT}} - \ell |}{\ell_{\text{OPT}}} \), where \( \ell_{\text{OPT}} \) and \( \ell \) denote the log-likelihoods of the optimal path and the path obtained via beam search, respectively.
Figure~\ref{Fig:Beam_width} illustrates the performance trends and corresponding relative decoding error under different beam widths. In terms of accuracy, when the beam width \( B \geq 64 \), the relative decoding error remains below 0.05\%. However, reducing \( B \) to 32 leads to a sharp error increase to 0.39\%, due to critical states being excluded from the narrower beam width. When the structural scale of a decoding task is held constant (i.e., fixed \( K \) and \( T \)), the impact of beam width on accuracy depends on the statistical characteristics of the task data, particularly the distribution of transition probabilities in \( \mathcal{A} \), \( \mathcal{B} \), and \( \pi \). In cases where transitions are dense or probabilities are uniformly distributed, optimal path states are more likely to fall outside the local top \( B \) at certain timesteps, resulting in accuracy degradation.
In contrast, the impact of beam width on the runtime and memory usage of \textsc{FLASH-BS Viterbi} is independent of the data distribution, depending solely on the state space size \( K \) and the chosen beam width \( B \). Both theoretical analysis and experimental results confirm that the runtime and memory usage are reduced to \( B/K \) of the original values. Overall, the performance–cost trade-off of beam search is influenced by domain-specific data distribution. We recommend analyzing this distribution before using small beam widths, and adjusting \( B \) accordingly. For the speech recognition task considered, \textsc{FLASH-BS Viterbi} achieves a 69.5× speedup with only a 0.05\% increase in error, demonstrating substantial gains in efficiency.

\subsection{Experiments on Resource-Constrained Devices}
To assess the efficiency and deployability of \textsc{FLASH} variants on resource-constrained devices, we evaluate two representative platforms: an FPGA accelerator and a Raspberry Pi 5. The FPGA experiments highlight the hardware friendliness and acceleration potential of our algorithms, while the Raspberry Pi experiments validate their efficiency on an embedded platform.
\subsubsection{FPGA-Based Acceleration}

Figures~\ref{Fig:hardware} show the decoding time of \textsc{FLASH Viterbi} and \textsc{FLASH-BS Viterbi} accelerators, along with the hardware speedup of \textsc{FLASH Viterbi}. The results indicate that speedup increases superlinearly with state space size \(K\), reaching 50.5$\times$ at \(K = 2048\), which demonstrates the efficiency of our hardware design. However, beam search disrupts sequential DDR memory access, leading to longer decoding time for \textsc{FLASH-BS Viterbi (HW)} compared with \textsc{FLASH Viterbi (HW)}. Nevertheless, its runtime remains comparable to that of the software implementation.

\begin{table} [tb]
\centering
\renewcommand{\arraystretch}{1.05} 
\setlength{\tabcolsep}{2pt} 
\caption{
Comparison of resource utilization and power consumption between \textbf{\textsc{FLASH-BS Viterbi}} (FLASH-BS) and \textbf{Reduced-Memory Viterbi} (RM-Viterbi).
The ``Decoding Unit'' refers only to the decoding logic; ``Entire Decoder'' includes the full accelerator with DDR interface modules. 
}
\resizebox{\linewidth}{!}{
\begin{tabular}{|c|l|ccc|ccc|c|}
\hline
\multirow{2}{*}{\textbf{Width}} & \multirow{2}{*}{\textbf{Algorithm}} 
& \multicolumn{3}{c|}{\textbf{Decoding Unit}} 
& \multicolumn{3}{c|}{\textbf{Entire Decoder}} 
& \multirow{2}{*}{\textbf{\makecell{Power\\(W)}}} \\
\cline{3-8}
& & BRAM & DSP & LUT+FF & BRAM & DSP & LUT+FF & \\
\hline
\multirow{2}{*}{32K} 
& FLASH-BS     & 173 & 0  & 13115 & 198.5 & 3  & 42445 & 1.835 \\
& RM-Viterbi   & 757 & 30 & 13927 & 764   & 30 & 35136 & 2.152 \\
\hline
\multirow{2}{*}{512} 
& FLASH-BS     & 7   & 0  & 13127 & 32.5  & 3  & 41954 & 1.737 \\
& RM-Viterbi   & 65  & 30 & 13927 & 72    & 30 & 34899 & 1.740 \\
\hline
\end{tabular}
}
\label{tab:Resource usage}
\end{table}

Table~\ref{tab:Resource usage} presents the resource utilization and power consumption of \textsc{FLASH-BS Viterbi}, compared with Reduced-Memory Viterbi (RM-Viterbi), at beam widths of 32K and 512. Power consumption was estimated using AMD’s Power Design Manager~\cite{amd_pdm_tool}. As the beam width decreases, \textsc{FLASH-BS Viterbi} achieves significant resource savings. Compared to RM-Viterbi, it incurs a slight LUT+FF overhead due to parallel DDR transfers, which can be mitigated by reducing the degree of transfer parallelism. For all other resource metrics, \textsc{FLASH-BS Viterbi} demonstrates substantial improvements. At a beam width of 32K, the Entire Decoder consumes only 26.0\% and the Decoding Unit 22.9\% of the BRAMs used by RM-Viterbi. This reduction in resource usage translates into improved power efficiency~\cite{Power_FPGA1,Power_FPGA2}, with \textsc{FLASH-BS Viterbi} consuming approximately 15\% less power than RM-Viterbi, demonstrating its hardware efficiency.

\subsubsection{Deployment on Raspberry Pi}
We further evaluate our algorithms on a resource-constrained device, the Raspberry Pi 5 (8GB). The beam width is set to half the state space. As shown in Figure~\ref{Fig:Raspberry}, although the limited computational capability of Raspberry Pi slightly reduces the benefits of multi-threading, \textsc{FLASH Viterbi\_P16} still matches or even outperforms the fastest baseline (Vanilla Viterbi). Moreover, beam search in \textsc{FLASH-BS Viterbi} further amplifies this advantage, highlighting the efficiency of \textsc{FLASH} variants on resource-constrained platforms.

\begin{figure}[t!]
    \vspace{0mm}
    \centering
    \includegraphics[width=1\columnwidth]{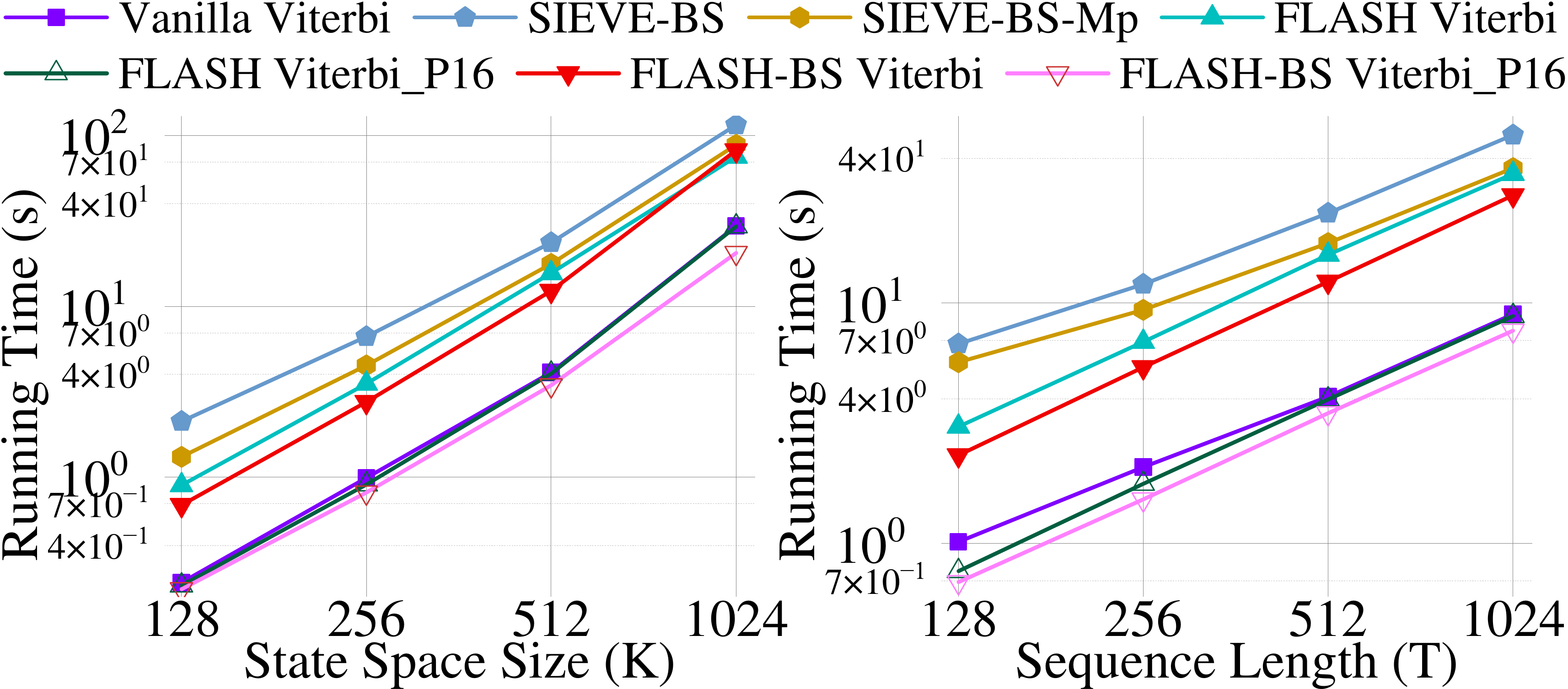}
    \caption{Decoding time of \textsc{FLASH} variants and baselines on the Raspberry Pi.}
    \label{Fig:Raspberry}
\end{figure}

\section{Conclusion}\label{sec:conclusion}
In this work, we propose \textsc{FLASH Viterbi} and \textsc{FLASH-BS Viterbi}, two decoding algorithms designed to address the performance and memory limitations of existing Viterbi implementations. Leveraging a non-recursive divide-and-conquer paradigm and a pruning and parallelization strategy, both methods enable highly parallel execution while reducing scheduling and memory overhead. The beam search variant further employs dynamic strategies and memory-efficient data structures to enhance adaptability under complex scenarios. To support deployment on edge platforms, we develop FPGA-based hardware accelerators for both algorithms, achieving substantial speedups with minimal resource consumption. Extensive experiments demonstrate the superiority of our methods in runtime and memory efficiency, while hardware results confirm their practicality in edge computing environments. 
Future work includes extending our framework to practical applications and exploring GPU-based acceleration.

\bibliographystyle{IEEEtran}  
\bibliography{ref}            

\end{document}